\documentclass[aps,pra,11pt,amsmath,amssymb,tightenlines,superscriptaddress,notitlepage,nofootinbib,onecolumn]{revtex4-2}
\pdfoutput=1

\usepackage[T1]{fontenc}
\usepackage[utf8]{inputenc}
\usepackage{amsthm,amsmath,amssymb,mathtools}
\usepackage{graphicx,booktabs,array}
\usepackage[dvipsnames,svgnames]{xcolor}

\let\mathscr\relax
\usepackage{newpxtext,newpxmath}
\usepackage[cal=cm]{mathalfa}
\usepackage{microtype}
\usepackage{placeins}
\usepackage{hyperref}
\hypersetup{colorlinks=true,linkcolor=blue!45!black,citecolor=blue!45!black,urlcolor=blue!45!black,
 pdftitle={Making the most of leftovers: Improved privacy amplification for quantum key distribution},
 pdfauthor={Matthew Simon Tan; Bartosz Regula; Marco Tomamichel}}
\newtheorem{theorem}{Theorem}
\newtheorem{proposition}[theorem]{Proposition}
\newtheorem{lemma}[theorem]{Lemma}
\theoremstyle{definition}
\newtheorem*{remark}{Remark}
\DeclareMathOperator{\Tr}{Tr}
\DeclareMathOperator{\Ree}{Re}

\newcommand{\id}{\mathbb{I}}
\newcommand{\allones}{\mathbb{J}}
\newcommand{\norm}[1]{\left\lVert #1\right\rVert}
\newcommand{\hmin}{H_{\min}}
\newcommand{\hmax}{H_{\max}}
\newcommand{\eps}{\varepsilon}
\newcommand{\pe}{\mathrm{pe}}

\newcommand{\pa}{\mathrm{pa}}
\newcommand{\pass}{\mathrm{pass}}
\newcommand{\Qtol}{Q_{\mathrm{tol}}}
\newcommand{\leakEC}{\mathrm{leak}_{\mathrm{ec}}}

\allowdisplaybreaks
\makeatletter
\def\fnum@figure{\textbf{Figure}\nobreakspace\textbf{\thefigure}}
\def\fnum@table{\textbf{Table}\nobreakspace\textbf{\thetable}}
\renewcommand{\@caption@fignum@sep}{: }
\renewcommand{\p@subsection}{\thesection.}
\renewcommand{\p@subsubsection}{\thesection.\thesubsection.}
\makeatother
\renewcommand{\thetable}{\arabic{table}}

\let\bar\relax
\newcommand{\lset}{\left\{\vphantom{\big|}}
\newcommand{\bar}{\;:\;}
\newcommand{\rset}{\vphantom{\big|}\right\}}

\let\texteq\relax
\newcommand{\texteq}[1]{\stackrel{\mathclap{\mbox{\scriptsize #1}}}{=}}

\let\oldproofname\proofname
\renewcommand{\proofname}{\rm\bf{\oldproofname}}

\makeatletter
\renewenvironment{proof}[1][\proofname]{\par
\pushQED{\qed}%
\normalfont \topsep6\p@\@plus6\p@\relax
\trivlist
\item\relax
{\bfseries  
#1\@addpunct{.}}\hspace\labelsep\ignorespaces 
}{%
\popQED\endtrivlist\@endpefalse
}
\makeatother

\begin{document}

\title{Making the most of leftovers:\texorpdfstring{\\}{ }Improved privacy amplification for quantum key distribution}

\author{Matthew Simon Tan}
\affiliation{Department of Physics, National University of Singapore, Singapore}

\author{Bartosz Regula}
\affiliation{Mathematical Quantum Information RIKEN Hakubi Research Team, RIKEN Pioneering Research Institute (PRI) and RIKEN Center for Quantum Computing (RQC), Japan}

\author{Marco Tomamichel}
\affiliation{Department of Electrical and Computer Engineering, National University of Singapore, Singapore}
\affiliation{Centre for Quantum Technologies, National University of Singapore, Singapore}

\begin{abstract}
The amount of secret key that can be obtained from a quantum key distribution run depends on both the physically observed error rates and the mathematical bounds used to certify security. For finite datasets, conservative bounds force users to discard a substantial fraction of the potentially available key. Here we further refine and extend the privacy amplification bounds achievable through the recent leftover hash lemma of Regula and Tomamichel~[\href{https://arxiv.org/abs/2603.04493}{arXiv:2603.04493}] and incorporate them into the security analysis of quantum key distribution based on entropic uncertainty relations. This improves on state-of-the-art key rates in finite-block regimes, certifying more secret key from the same experimental data without changes to the protocol, and outperforming techniques based on entropy accumulation. 
The results illustrate how sharper mathematical estimates can directly increase the usable output of a quantum communication system.
\end{abstract}

\maketitle

\section{Introduction}

The final output of a quantum key distribution experiment~\cite{bb84,ekert1991} is a string of secret bits. Its length is determined not only by how many signals reach the receiver, or by how often their measurement outcomes disagree, but also by what can be proved about the information available to an adversary. Under the assumptions of a security proof, the observed data constrain that information and determine how much of the shared raw string can safely be retained. This makes the quality of a mathematical bound an operational resource: a sharper bound can turn the same experimental data into a longer certified key, at the same security level. This connection between a proof and a usable output is particularly direct in quantum cryptography, where security must be certified by information-theoretic arguments that hold against an adversary with quantum side information and must remain meaningful when the key is subsequently used~\cite{portmann2022}.

Progress in quantum key distribution spans both security theory and experimental implementation~\cite{pirandola2020,xu2020}. 
The finite size of the data set is a central consideration here: asymptotic key rates describe arbitrarily long experiments, but any implementation can only process a finite block of signals. Statistical fluctuations, information revealed during error correction, and the conversion of partially secret data into an almost uniform key all incur finite-size costs, the accurate estimation of which then becomes vital.

Short blocks are especially important when data collection is constrained by the communication opportunity. As a representative example, satellite experiments have established the feasibility of distributing quantum keys over long free-space links~\cite{liao2017,yin2020}, but for low-Earth-orbit links, finite visibility windows and optical losses restrict the number of detections available during an individual pass. Pooling observations over longer periods can increase the block size, but it changes the latency and operational requirements. Studies of satellite key generation have therefore identified finite-key analysis as a significant determinant of performance~\cite{lim2021,sidhu2022}.

Improving the statistical analysis already has demonstrable consequences in this setting. Lim et al.~\cite{lim2021} tightened the sampling analysis for small blocks and examined its implications for space-based communication. More recently, Mannalath, Zapatero and Curty~\cite{mannalath2025} developed sharper finite-statistics tools and showed that exact evaluation of hypergeometric probabilities can substantially improve upon general-purpose tail bounds. These results establish that a portion of the finite-key penalty reflects the conservatism of a proof rather than an unavoidable limitation of the experiment. They also motivate examining other steps in the security analysis where an avoidable relaxation may discard usable key.

Privacy amplification is a particularly important step of establishing the security of a QKD protocol, but at the same time one whose mathematical analysis left room for improvement. In this step, Alice and Bob apply a publicly chosen hash function to their reconciled strings to obtain a shorter key that is nearly uniform and independent of the adversary. A central result here is the leftover hash lemma (LHL)~\cite{bennett1995,tssr2011}, which allows for an estimation of the output length using the entropy of the input. In the quantitative analysis of privacy amplification, an important role is played by the smooth min-entropy~\cite{rennerwolf2005,rennerkoenig2005,renner2005}, and Renner's thesis~\cite{renner2005} already described this characterisation as essentially optimal: universal hashing provides an achievable key length via the LHL, while converse bounds limit the amount of secure randomness that any family of hash functions can extract, and both of these can be expressed in terms of the smooth min-entropy. However, the achievable and converse bounds involve security-dependent corrections and different relations between the smoothing and security parameters, and even minuscule gaps in such parameters can have major consequences at finite block lengths and stringent security levels. 

The room for improvement in the early bounds became particularly apparent through a different approach to privacy amplification. Dupuis~\cite{dupuis2023} derived a bound directly in terms of sandwiched R\'enyi entropies, avoiding an intermediate conversion to smooth min-entropy. This allows a security analysis to retain its R\'enyi entropy description through to the final key-length bound. Recent developments in entropy accumulation~\cite{dupuis2020,metger2023,arqand2025} exploited this route and obtained substantially improved finite-key performance. In particular, the R\'enyi entropy accumulation theorem (REAT) of Arqand, Hahn and Tan~\cite{arqand2025}, together with subsequent applications~\cite{kamin2025eat,kamin2025renyi}, demonstrates the practical advantage of avoiding losses in the conventional smooth min-entropy conversion. More recently, Goh and Tan~\cite{gohtan2026} implemented these key rate computations at short block lengths, using improved parameter choices. These gains suggested that the original bound in terms of smooth min-entropy was not as tight as it could have been.

Indeed, recently some of us~\cite{rt2026} showed that much tighter privacy amplification bounds can be obtained by revisiting how smoothing is used in the quantum setting. The work introduced measured smooth entropies and established a strengthened LHL, which improved on previous smooth min-entropy estimates and recovered the tight R\'enyi entropy bounds of Dupuis~\cite{dupuis2023}. This raises a natural question for finite-key quantum cryptography: how much of this improvement can be retained when the available entropy guarantee is expressed in terms of the conventional smooth min-entropy that features in standard security proofs?
This is particularly important because bounds in terms of the smooth min-entropy can be plugged directly into many existing proof approaches, without the need to invoke R\'enyi entropy accumulation techniques.  
Here, as part of our technical contribution, we further sharpen the connection between measured smooth entropies of~\cite{rt2026} and the familiar smooth min-entropy, deriving an improved bound on the amount of key that can be certified from such a guarantee in terms of min-entropy. We furthermore extend the scope of applicability of the underlying techniques by generalising them to all 2-universal hash functions.

To quantify the resulting gain in quantum key distribution, we build on the finite-key security analyses developed using smooth entropies, entropic uncertainty relations (EURs) and complementary techniques~\cite{renner2005,tomamichel2011eur,tomamichel2012,hayashi2012}. These methods provide explicit security guarantees against general attacks for experimentally relevant block sizes. Building on the earlier finite-key analysis of Ref.~\cite{tomamichel2012}, we incorporate our improved privacy amplification bound into the framework of Tomamichel and Leverrier~\cite{tl2017}, which uses an EUR~\cite{tomamichel2011eur} to obtain a smooth min-entropy guarantee and provides a complete analysis for a family of entanglement-based and prepare-and-measure protocols, including variants of BB84~\cite{bb84}. Its explicit treatment of the security parameters, sampling rule, error correction leakage and final key length makes it particularly suitable for quantifying the effect of a sharper privacy amplification bound. We combine the new hashing bound with additional improvements that can be obtained from the exact evaluation of the relevant sampling failure probability. Importantly, the quantum operations and acceptance rule remain unchanged --- the improvement lies in the length of key certified from the same data. 

We compare the resulting key lengths with the refined REAT rates of Goh and Tan~\cite{gohtan2026}, which represent the state-of-the-art finite-key bounds obtained from entropy accumulation approaches. This comparison tests the practical value of the improved smooth-entropy analysis against a modern alternative, beyond an improvement over the older uncertainty relation bound alone. 
Our results show an advantage of our approach in the regimes studied, particularly at short block lengths. The improvement thus demonstrates how revisiting a mathematical step that was understood to be essentially optimal can nevertheless yield more usable key from a finite quantum communication experiment.

\section{Main results}
\label{sec:main}

\subsection{Refined leftover hashing for privacy amplification}

Our main technical result is a sharper guarantee for the amount of secret key that can be extracted from a given smooth min-entropy estimate. It improves the final conversion from an entropy guarantee to a key length, so that the same experimental data can support a longer key at the same security level. We first state this guarantee and explain the source of the gain, then show its effect in the finite-block comparisons of Figure~\ref{fig:rates} and Table~\ref{tab:rates}.

Let $X$ be a classical raw string and $E$ the adversary's quantum side information, jointly described by a possibly subnormalised state $\rho_{XE}$. The quantity $\hmin^\eps(X|E)_\rho$ is the conditional smooth min-entropy: it measures the randomness available against $E$ after allowing a controlled approximation of size $\eps$, measured in terms of purified distance (see Appendix~\ref{app:conventions} for details). Let $\mathcal H=\{h_s:\mathcal X\to\mathcal Z\}_{s\in\mathcal S}$ be the hash family. An independent uniform public seed $S$ is used to select $h=h_S$, giving output $Z=h(X)$. We write
\begin{equation}
 \Delta(\rho,\mathcal H)
 \coloneqq\mathbb E_S\left[\frac12\norm{\omega_{ZE}^{S}-\pi_Z\otimes\rho_E}_1\right],
 \label{eq:secrecy}
\end{equation}
where $\omega_{ZE}^{s}$ is the state after hashing with seed value $s$ and $\pi_Z$ is uniform on $\mathcal Z$. Thus $\Delta$ measures how far the output is from a uniform key independent of the adversary, including knowledge of the public seed. 

The tightest formulation of our result applies to $2^*$-universal families of hash functions, where any two distinct inputs collide with probability of exactly $1/|\mathcal Z|$. This is the same as the assumption made in many related works~\cite{tl2017,dupuis2023,rt2026}, and the commonly used hashing method based on uniformly sampling binary Toeplitz matrices forms one such family. However, a slightly relaxed formulation of our result in fact applies more broadly: in Appendix~\ref{app:proofs} we extend Theorem~\ref{thm:pa} to more general, 2-universal hash functions, where only an upper bound on the collision probability is known.

\begin{theorem}[Privacy amplification]
\label{thm:pa}
Let $\rho_{XE}$ be a finite-dimensional subnormalised classical--quantum state and let $0\leq\eps<\sqrt{\Tr\rho_{XE}}$. For any $2^*$-universal family of hash functions,
\begin{equation}
 \Delta(\rho,\mathcal H)
 \leq \eps^2+\frac12\sqrt{\big(|\mathcal Z|-1\big)\,\kappa(\eps)\,
             2^{-\hmin^\eps(X|E)_\rho}}
 \label{eq:newlhl}
\end{equation}
where one may take
\begin{equation}
 \kappa(\eps)=
 \begin{cases}
  (1-\eps^2)(1+3\eps^2),&0\leq\eps\leq1/\sqrt3,\\
  4/3,&1/\sqrt3<\eps<1.
 \end{cases}
 \label{eq:kappa}
\end{equation}
In particular, $1\leq\kappa(\eps)\leq4/3$ and $\kappa(\eps)=1+2\eps^2+O(\eps^4)$ as $\eps\to0$.
\end{theorem}

Compared with the smooth min-entropy bound in~\cite[Corollary~15]{rt2026}, Theorem~\ref{thm:pa} retains the quadratic smoothing cost $\eps^2$ while improving the dependence on $|\mathcal Z|\,2^{-\hmin^\eps(X|E)_\rho}$ from a cube root to a square root.

The conventional bound in terms of purified distance smooth min-entropy, due to Tomamichel, Schaffner, Smith and Renner~\cite{tssr2011}, is
\begin{equation}
 \Delta(\rho,\mathcal H)
 \leq 2\eps+\frac12\sqrt{|\mathcal Z|\,2^{-\hmin^\eps(X|E)_\rho}}.
 \label{eq:oldlhl}
\end{equation}
The essential gain is the replacement $2\eps\mapsto\eps^2$. At a fixed secrecy target, this permits a much larger smoothing radius and hence avoids an overly pessimistic entropy estimate. The competing penalty term $\kappa(\eps)\geq1$ must be traded against the gain from the smaller smoothing error --- its multiplicative contribution to the entropy term corresponds to the loss of $\log_2\kappa(\eps)$ bits. However, this penalty is at most $\log_2(4/3)\simeq0.415$ bits and is essentially zero in the small-$\eps$ regime in which we primarily apply the result.

The practical significance becomes clear when the entropy estimate comes from a sampling test. In the security analysis of key distribution protocols, an excluded event of probability $p$ leads to a purified distance of $\sqrt p$. The conventional hashing bound would then incur a cost of $2 \sqrt p$, whereas our achievability result only incurs $p$. At a fixed security target, this allows a larger statistical failure budget and hence an improved estimate of the unobserved errors. 

Another way to understand the advantage is in terms of the scaling of the extractable randomness.  
For a binary key of length $\ell$, the output alphabet has size $|\mathcal Z|=2^\ell$, and we can solve Theorem~\ref{thm:pa} for $\ell$. Take then a state $\rho_{XE}$ and let $\ell_{\mathrm{ext}}^{\eps_{\mathrm{sec}}}(X|E)_\rho$ denote the largest integer key length obtainable with secrecy error at most $\eps_{\mathrm{sec}}$, allowing arbitrary hash families with an independent public seed. The achievability result of Theorem~\ref{thm:pa} gives
\begin{equation}
 \ell_{\mathrm{ext}}^{\eps_{\mathrm{sec}}}(X|E)_\rho
 \geq
 \left\lfloor
 \hmin^{\sqrt{\eps_{\mathrm{sec}}-\delta}}(X|E)_\rho
 -\log_2\frac{1}{3\delta^2}
 \right\rfloor
 \label{eq:keylower}
\end{equation}
for any $0<\delta<\eps_{\mathrm{sec}}$. 
The older incarnations of the leftover hash lemma would have instead indicated a bound in terms of $\hmin^{\eps_{\mathrm{sec}}/2-\delta}(X|E)_\rho$, which has a markedly worse dependence on the target secrecy error $\eps_{\mathrm{sec}}$. 
Importantly, the scaling of~\eqref{eq:keylower} is indeed optimal: the converse bound of~\cite[Proposition~21]{rt2026} implies that
\begin{equation}
 \ell_{\mathrm{ext}}^{\eps_{\mathrm{sec}}}(X|E)_\rho
 \leq\hmin^{\sqrt{\eps_{\mathrm{sec}}}}(X|E)_\rho
       +\log_2\frac1{1-\eps_{\mathrm{sec}}},
 \label{eq:keyupper}
\end{equation}
and this applies to general hash functions, without even assuming that the hash family needs to be universal. Together, the achievable and converse bounds thus precisely identify the optimal dependence of the smooth min-entropy parameter on the target secrecy error as $\sqrt{\eps_{\mathrm{sec}}}$.
Here we note that, while the scaling itself already follows from the findings of~\cite{rt2026}, our Theorem~\ref{thm:pa} additionally improves the dependence on the additive fudge term $\delta$ in~\eqref{eq:keylower}, which is needed to obtain the stronger bounds on key rates in this work.


\subsection{Security analysis}

The security of QKD can be established through several closely related
approaches. These include the pioneering proofs of
Mayers~\cite{mayers2001}, the entanglement-distillation approach of
Shor and Preskill~\cite{shorpreskill2000}, the complementarity approach of
Koashi~\cite{koashi2009}, Renner's
reduction to collective attacks via smooth entropy~\cite{renner2005}, entropic uncertainty
relations~\cite{tomamichel2011eur}, and entropy
accumulation~\cite{dupuis2020,metger2023}. 
Here we adopt the entropic uncertainty relation framework of Tomamichel and Leverrier~\cite{tl2017}, which provides a particularly
elementary finite-key security proof for
BB84-type protocols. Our main comparison point are entropy accumulation methods, since they currently provide the most competitive key rates~\cite{arqand2025,gohtan2026}.

We retain the protocol structure and
uncertainty relation argument of~\cite{tl2017}, while sharpening two ingredients:
we replace the conventional smooth min-entropy leftover hashing lemma
with Theorem~\ref{thm:pa}, and the analytical sampling bound used in
parameter estimation with an exact hypergeometric probability.
For the numerical completeness analysis, we also evaluate the honest
test-abort probability using the exact binomial tail. The accounting
for information disclosed during reconciliation and verification
otherwise follows the original framework, and we refer to~\cite{tl2017} for details.

We quickly introduce the relevant concepts and variables for the security analysis here. Of $m$ rounds, a uniformly selected sample of $k$ is used to estimate the error rate, leaving $n=m-k$ key rounds. The protocol continues only if the test error rate is below $\Qtol$. Alice then sends Bob an error correction message of at most $\leakEC$ bits and a verification tag of length $t$ over the authenticated public channel. Bob reconciles his string and checks the tag. If verification succeeds, they apply privacy amplification to produce keys of length $\ell$.

\begin{figure}[!t]
 \centering
 \includegraphics{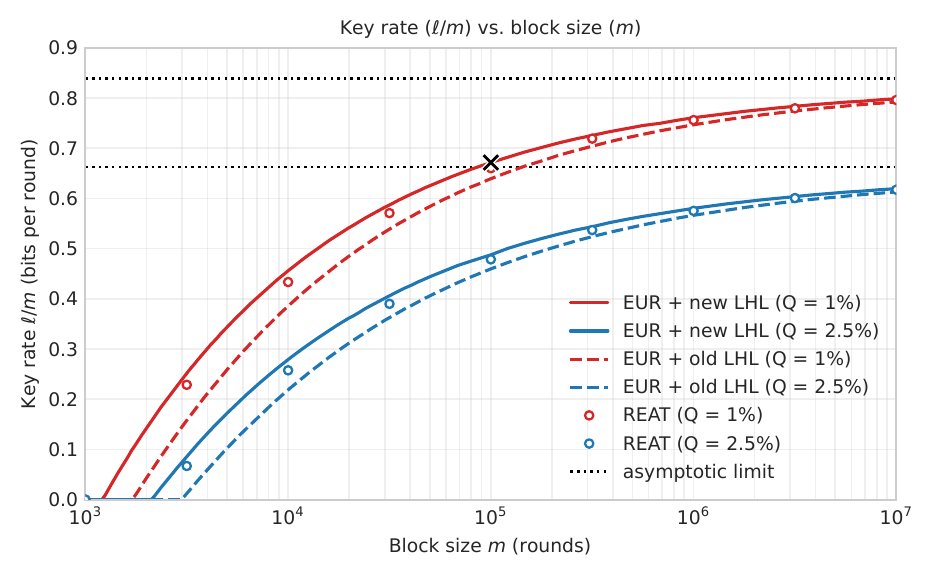}
 \caption{\textbf{Finite-block key rates.} Secret key bits per round, $\ell/m$, are plotted for quantum bit error rates $Q=1\%$ and $2.5\%$. The entropic uncertainty relation (EUR) curves compare the new and old leftover hash lemmas (LHL) with exact sampling. Their inputs are ideal BB84 overlap $c=1/2$, $\eps_{\mathrm{tot}}=10^{-10}$ and $\eps_{\mathrm{com}}=0.01$. Circles show the corresponding key rates from the refined analysis of the R\'enyi entropic accumulation theorem~\cite{arqand2025} due to Goh and Tan~\cite{gohtan2026}. For each EUR curve, the protocol parameters $k,n,t,\Qtol$ are optimised to maximise $\ell$. Table~\ref{tab:example-parameters} gives example parameter values for the crossed point. See Section~\ref{sec:methods} and Appendix~\ref{app:numerics} for details on the numerical analysis.}
 \label{fig:rates}
\end{figure}

The precise notion of security we employ here is in the composable correctness-and-secrecy sense of~\cite{portmann2022} and~\cite{tl2017}. The security bound has three contributions: verification, parameter estimation and privacy amplification. Verification bounds the probability of accepting unequal keys by $\eps_{\mathrm{cor}}(t)=2^{-t}$. For $0<\nu<1/2-\Qtol$, let $\eps_\pe(\nu)$ be the worst-case joint probability that the test passes and the error rate relevant to the remaining key rounds is at least $\Qtol+\nu$. Here $\nu$ is the extra error rate above the test threshold $\Qtol$. (More details can be found in Section~\ref{sec:methods} and Appendix~\ref{app:sampling}.) At fixed output length $\ell$, let $\eps_\pa(t,\nu)$ denote the privacy amplification contribution, given by the square-root term in equation~\eqref{eq:qkd}. The three contributions must satisfy $\eps_{\mathrm{cor}}(t)+\eps_\pe(\nu)+\eps_\pa(t,\nu)\leq\eps_{\mathrm{tot}}$; the last two terms together bound secrecy.

Let $\pass$ denote acceptance of both parameter estimation and
verification. Following the entropy argument leading to
Eq.~(102) of~\cite{tl2017},
but replacing its sampling estimate with the exact failure probability
$\eps_{\mathrm{pe}}(\nu)$, we obtain, whenever
$\Pr[\pass]>\eps_{\mathrm{pe}}(\nu)$,
\begin{equation}
 \hmin^{\sqrt{\eps_{\mathrm{pe}}(\nu)}}(X\mid E\wedge\pass)
 \geq n[1-h(\Qtol+\nu)]-\leakEC-t
 \eqqcolon H_*(\nu,t).
 \label{eq:hstar}
\end{equation}
Here $\wedge\pass$ denotes restriction to the subnormalised accepted
branch, and $E$ includes the adversary's quantum information and the
public transcript, excluding the independent privacy-amplification
seed. The uncertainty relation supplies the term
$n[1-h(\Qtol+\nu)]$, while reconciliation and verification reduce
the entropy by at most $\leakEC+t$ bits.
Combining this estimate with Theorem~\ref{thm:pa} gives the following
security criterion.

\begin{theorem}[Finite-key security]
\label{thm:qkd}
Fix the round counts $m,k,n$ with $m=k+n$ and the test threshold $\Qtol$. Under the protocol assumptions of~\cite{tl2017}, specialised to ideal BB84 measurements and with a $2^*$-universal family for privacy amplification, an $\eps_{\mathrm{tot}}$-secure final key of length $\ell$ can be extracted whenever
\begin{equation}
 \inf\left\{\eps_{\mathrm{cor}}(t)+\eps_\pe(\nu)+
 \sqrt{\frac{2^{\ell-H_*(\nu,t)}}{3}}\;:\;0<\nu<\frac12-\Qtol,\ t\in\mathbb N\right\}
 \leq\eps_{\mathrm{tot}},
 \label{eq:qkd}
\end{equation}
with $H_*(\nu,t)$ defined in~\eqref{eq:hstar}. 
\end{theorem}

For valid choices of $\nu,t$ with $\eps_{\mathrm{cor}}(t)+\eps_\pe(\nu)<\eps_{\mathrm{tot}}$, Theorem~\ref{thm:qkd} gives the certified integer key length
\begin{equation}
 \ell_{\mathrm{cert}}=
 \left\lfloor H_*(\nu,t)+2\log_2\!\bigl[\eps_{\mathrm{tot}}-\eps_{\mathrm{cor}}(t)-\eps_\pe(\nu)\bigr]+\log_2 3\right\rfloor.
 \label{eq:keylength}
\end{equation}
The sampling-failure probability $\eps_\pe(\nu)$ also depends on $k$, $n=m-k$, and $\Qtol$; we have suppressed these dependencies because these protocol parameters are fixed in Theorem~\ref{thm:qkd}.

The above notion of security only guarantees that the probability that the key is insecure and the protocol does not abort is small. Since a protocol that always aborts would meet this criterion, one considers a completeness parameter. This parameter $\eps_{\rm com}$ upper bounds the probability the protocol aborts when the eavesdropper is absent. This is referred to as the `honest implementation'. The threshold $\Qtol$ and the error correction code parameters are then chosen to satisfy this condition.

For the EUR curves, we numerically search over $k$ and $\nu$, with $n=m-k$, $\Qtol$ set by completeness, and $t$ optimised at each candidate. Changing $\nu$ only changes the analysis of a fixed protocol, while the selected protocol parameters may differ between curves. The details of the optimisation are described in Section~\ref{sec:methods} (Methods) and Appendix~\ref{app:numerics}.

Figure~\ref{fig:rates} shows that the improved LHL has its largest relative effect at short block lengths. In this regime, the uncertainty in the unobserved errors and the cost of turning the remaining entropy into a secure key consume a substantial fraction of the available randomness. The new bound reduces the secrecy cost associated with the sampling-failure probability, allowing a less pessimistic entropy estimate at the same security target. The gain therefore comes both from a tighter conversion of entropy into key and from the greater sampling failure probability that the security budget can accommodate. As the block size grows, these finite-size costs become smaller relative to the block, and the advantage diminishes.

The comparison in Fig.~\ref{fig:rates} and Table~\ref{tab:rates} shows that this hashing improvement remains valuable even after the sampling analysis has been sharpened. We use exact sampling for both the old and the new LHL EUR calculations --- their difference reflects the improved hashing bound and the resulting choice of protocol parameters. The two refinements thus address distinct sources of conservatism in the security proof. The comparison with the refined REAT results computed in Goh and Tan~\cite{gohtan2026} also shows that a smooth min-entropy analysis can remain competitive at short block lengths. The older hashing bound can obscure this potential by imposing an unnecessarily large cost when the statistical uncertainty is substantial.

The benefit requires no changes to the implementation or hashing construction: the protocol retains the same sequence of operations, while the sharper analysis permits more of the available randomness to be retained as key.

Short-block performance matters when the amount of data cannot readily be increased. Satellite links provide a motivating example: finite communication windows and optical losses limit the detections available within a pass~\cite{lim2021,sidhu2022}. Here $m$ counts the rounds retained for testing and key generation after detection and sifting, rather than transmitted pulses. A sharper finite-key bound can make such limited data more useful without requiring a longer acquisition period. Figure~\ref{fig:satellite} examines this regime within the stated BB84 and leakage models, showing both the absolute key output and $M_0$, the smallest block size certified to yield a secret bit on the specified parameter lattice. Since authentication costs are excluded, $M_0$ marks positive key output under the assumed authenticated channel and does not by itself establish net key expansion. Their dependence on the bit error rate and the security target makes clear why short-block gains matter: they affect how much data is needed before key extraction becomes possible, as well as how much key a finite data set can yield.

\begin{table}[!htbp]
 \centering
 \begin{tabular}{ccccc}
 \toprule
 $Q$ & $m$ & EUR + old LHL & EUR + new LHL & REAT\\
 \midrule
 $1\%$ & $10^4$ & $0.385$ & $0.456$ & $0.433$\\
         & $10^5$ & $0.639$ & $0.671$ & $0.660$\\
         & $10^6$ & $0.747$ & $0.761$ & $0.756$\\
 $2.5\%$ & $10^4$ & $0.219$ & $0.279$ & $0.258$\\
         & $10^5$ & $0.460$ & $0.488$ & $0.479$\\
         & $10^6$ & $0.566$ & $0.580$ & $0.575$\\
 \bottomrule
 \end{tabular}
 \caption{Selected rates $\ell/m$ from the same datasets as Fig.~\ref{fig:rates}, rounded to three decimal places. In order to ensure a fair comparison, both EUR columns include the improved tail analysis through exact hypergeometric sampling; they differ in the leftover hash lemma used, and separately optimise the protocol parameters using the same security and completeness targets and reconciliation model. The old LHL column therefore does not reproduce the original analysis of~\cite{tl2017} with its analytical tail bound. The REAT values are taken from the code  adapted from ~\cite{gohtan2026}, as discussed in Section \ref{sec:methods-numerics}.}
 \label{tab:rates}
\end{table}

\begin{figure}[!htbp]
 \centering
 \includegraphics{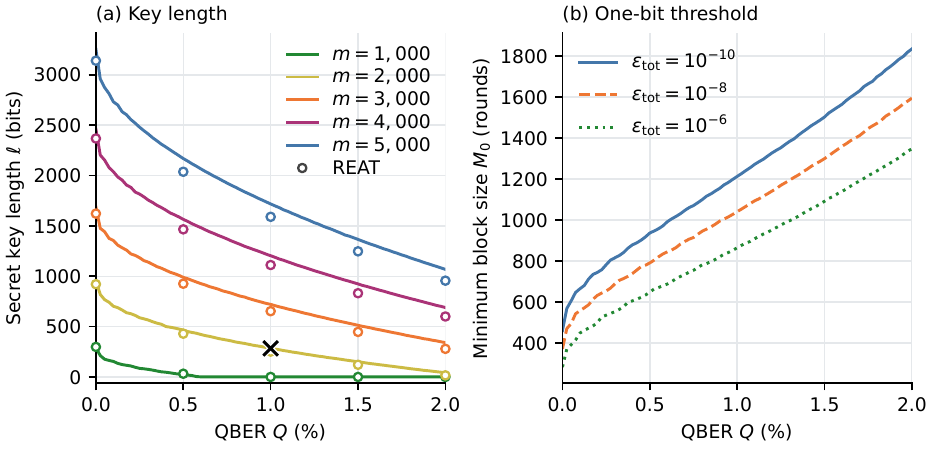}
\caption{\textbf{Key extraction from short blocks.}
(a) Secret key length for $m\in\{1000,2000,3000,4000,5000\}$
at $\eps_{\mathrm{tot}}=10^{-10}$.
Lines show EUR + new LHL with exact sampling; circles show the
REAT results~\cite{gohtan2026}, with negative rates displayed at zero.
(b) Minimum block size $M_0$ yielding one secret bit, for
$\eps_{\mathrm{tot}}\in\{10^{-10},10^{-8},10^{-6}\}$.
Both panels' EUR calculations use ideal BB84 overlap $c=1/2$,
$\eps_{\mathrm{com}}=0.01$ and assumed leakage coefficients
$(\xi_1,\xi_2)=(1.13,3.82)$ at every quantum bit error rate (QBER).
Panel (a) uses $\kappa=4/3$; (b) retains the full bound of
Theorem~\ref{thm:pa}.
We optimise $k,n,t,\Qtol$ to maximise $\ell$ at fixed $m$ in (a),
and minimise $m$ with $\ell\geq1$ in (b).
The cross identifies the example in Table~\ref{tab:example-parameters}.
See Section~\ref{sec:methods} and Appendix~\ref{app:numerics}.}
\label{fig:satellite}
\end{figure}
\FloatBarrier

\section{Methods}
\label{sec:methods}

We first explain how the improved entropy comparison strengthens leftover hashing, then describe the exact tails used for completeness and parameter estimation. Finally, we specify the physical model and optimisation choices behind the numerical results. The complete mathematical and computational details are provided in the Appendix.

\subsection{Leftover hashing and smooth min-entropy}\label{sec:methods-hashing}

Privacy amplification is naturally governed by quadratic quantities: averaging over a hash family brings in the probability that two inputs collide, which is why families of hash functions that control this collision probability are often most useful. Classically, this leads to a natural correspondence with the collision divergence $D_2(p\|q)=\log_2\sum_x p_x^2/q_x$ for a distribution $p$ and reference $q$. Smoothing permits a small amount of probability to be discarded before evaluating this quantity, with the discarded mass incurring a linear, rather than quadratic, cost. For quantum side information, \cite{rt2026} defined a measured smooth collision divergence $D_2^{\eta,\,\mathbb{M}}$: it is based on smoothing the outcome distribution of each measurement by at most $\eta$ in total variation (trace) distance, then taking the supremum over measurements. The optimisation over measurements here is a technical property and is not connected with the operational protocols that the quantity aims to describe. Indeed, the optimisation can be expressed without invoking measurements --- however, unlike conventional smoothing notions, this then leads to a smoothing over Hermitian operators rather than just quantum states.

For any independent uniform public seed family of 2*-universal hash functions $\mathcal H=\{h_s:\mathcal X\to\mathcal Z\}_{s\in\mathcal S}$, i.e.\ one such that the probability of input collisions $\Pr_S[h_S(x) = h_S(x')]$ uniformly equals $1/|\mathcal Z|$ for all $x\neq x'$, the leftover hashing bound of~\cite[Theorem~14]{rt2026} takes the form
\begin{equation}
 \Delta(\rho,\mathcal H)\leq\eta+
 \frac12\sqrt{(|\mathcal Z|-1)\,
 2^{D_2^{\eta,\mathbb{M}}(\rho_{XE}\|\id_X\otimes\sigma_E)}}
 \label{eq:measuredlhl}
\end{equation}
for any $0\leq\eta<\Tr\rho$ and any state $\sigma_E$. The first term is the linear smoothing cost, and the second is the quadratic collision cost that directly bounds the size of the output alphabet. The precise divergence and operator representations are given in Appendix~\ref{app:measured}.

To connect this result to existing QKD entropy estimates, we need a bound in terms of $\hmin^\eps$. We use purified distance smoothing~\cite{tomamichel2010duality}: for normalised states, $P(\rho,\rho')=\sqrt{1-F(\rho,\rho')^2}$, where $F(\rho,\rho')=\norm{\sqrt\rho\sqrt{\rho'}}_1$ is the root fidelity. 
In the security analysis framework that we employ here, it becomes necessary to accommodate accepted protocol branches that involve subnormalised states with $\Tr \rho < 1$; for this, the fidelity and purified distance are naturally extended to generalised variants~\cite{tomamichel2010duality,tomamichelbook}. 
The smooth min-entropy $\hmin^\eps(X|E)_\rho$ is defined as the supremum of $\hmin(X|E)_{\rho'}$ over all subnormalised states $\rho'_{XE}$ with $P(\rho,\rho') \leq \eps$. 
Our key contribution is then Proposition~\ref{prop:entropycomparison}, proved in Appendix~\ref{app:comparison}, which gives
\begin{equation}\begin{aligned}
 \inf_{\sigma_E\in\mathcal S(E)}
 D_2^{\eps^2,\,\mathbb{M}}(\rho_{XE}\|\id_X\otimes\sigma_E)
 &\leq-\hmin^\eps(X|E)_\rho+\log_2\kappa(\eps)\\
 &\leq -\hmin^\eps(X|E)_\rho+\log_2 \frac43.
 \label{eq:entropycomparison}
\end{aligned}\end{equation}
Inserting this comparison into equation~\eqref{eq:measuredlhl} with $\eta=\eps^2$ allows us to retain the square-root entropy term, giving exactly Theorem~\ref{thm:pa}. This improves the relaxation to smooth min-entropy in Ref.~\cite[Corollary~15]{rt2026}, which only gives a cube-root entropy term. 
Our proof approach exploits the variational form of the measured smooth collision divergence $D_2^{\eta,\,\mathbb{M}}$ to directly connect it with the purified distance smoothing needed for $H_{\min}^\eps$, instead of going through intermediate one-shot divergence inequalities as in~\cite{rt2026}.
The coefficient $\kappa(\eps)$ comes from the fidelity estimate in our proof and tends to one for small $\eps$. 

Another contribution we make is to extend the leftover hashing approach of~\cite{rt2026} beyond 2*-universal hash functions. Although some commonly employed hash functions are indeed 2*-universal, it is known that relaxing this assumption slightly can lead to reductions in the seed size needed to implement the hashing~\cite{stinson_1994,tssr2011}. Precisely, in Appendix~\ref{app:hash} we show that only assuming an upper bound on the collision probability, $\Pr_S[h_S(x)=h_S(x')]\leq\frac{1+\mu}{|\mathcal Z|}$ for $x \neq x'$, leads to an achievability result of the form
\begin{equation}
 \Delta(\rho,\mathcal H)\leq\eps^2+
 \frac12\sqrt{2 \big[|\mathcal Z|-1 + \mu(|\mathcal X|-1)\big]\,
 \kappa(\eps)\, 2^{-\hmin^\eps(X|E)_\rho}}.
\end{equation}
This shows in particular that using our approach with the ordinary notion of 2-universal hash functions, where $\mu = 0$, guarantees a security bound that is only a single bit worse than the tight 2*-universal statement in Theorem~\ref{thm:pa}.

\subsection{Exact tails for parameter estimation and completeness}\label{sec:methods-tails}

\begin{figure}[!htbp]
 \centering
 \includegraphics[width=0.75\linewidth]{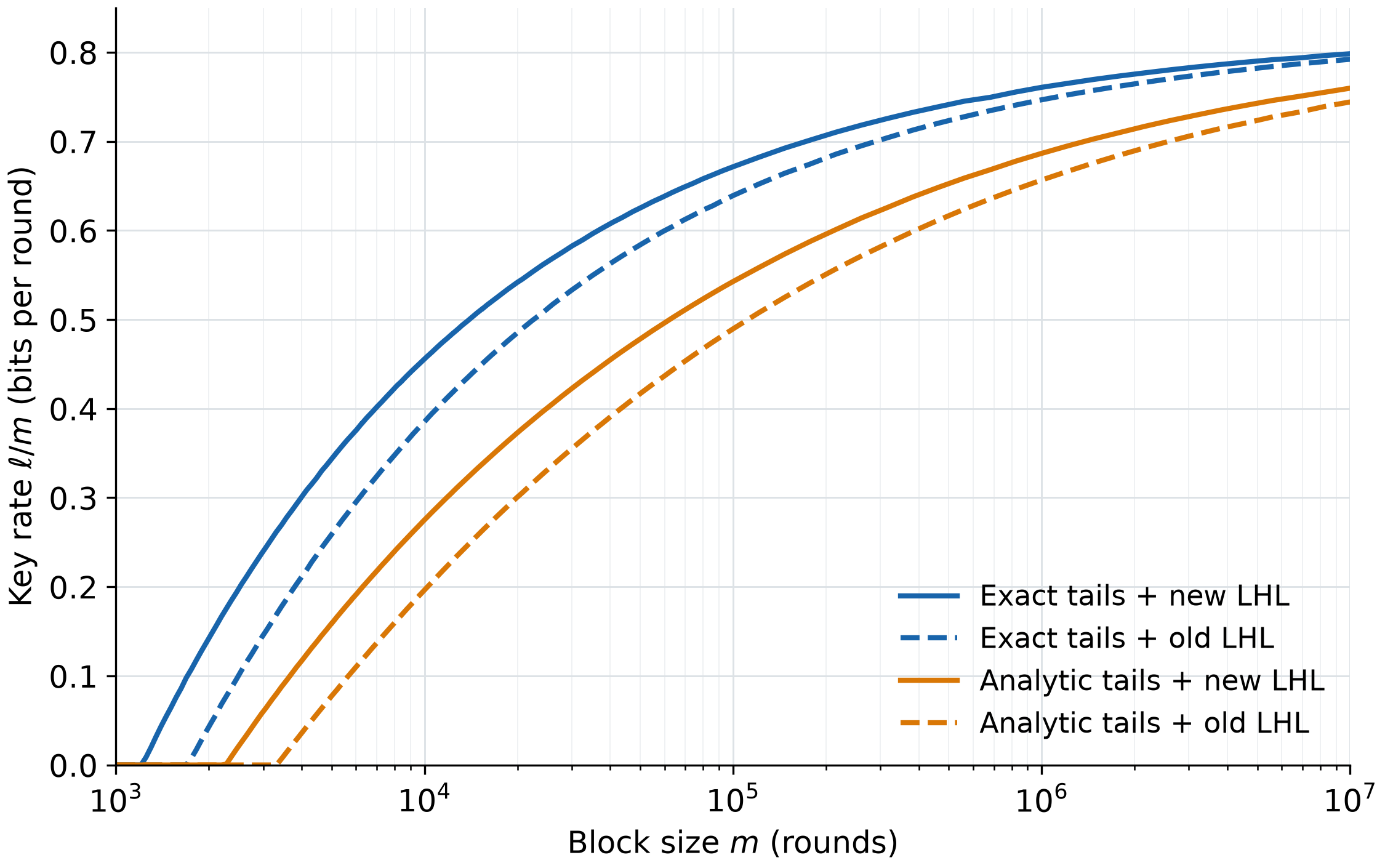}
 \caption{\textbf{Contributions of exact tails and improved hashing.} Secret key bits per round at $Q=1\%$, comparing exact and analytic tails with the new and old LHLs. Inputs are ideal BB84 overlap $c=1/2$, $\eps_{\mathrm{tot}}=10^{-10}$ and $\eps_{\mathrm{com}}=0.01$; the new LHL uses $\kappa=4/3$. For each choice of tail bounds and LHL, the protocol parameters $k,n,t,\Qtol$ are optimised to maximise $\ell$. Table~\ref{tab:example-parameters} gives example parameter values for the new LHL with exact sampling. See Section~\ref{sec:methods} and Appendix~\ref{app:numerics} for details on the numerical analysis.}
 \label{fig:sampling-comparison}
\end{figure}

During the parameter estimation stage of the protocol, Alice announces her test-round outcomes and Bob compares them with his own. If the fraction of errors is at least $\Qtol$, the protocol aborts. However, even in the acceptance branch, the fraction of errors in the rounds used for key generation may still exceed $\Qtol$ by some $\nu \in (0,\frac{1}{2}-Q_{\mathrm{tol}})$. We refer to the probability of this event as the sampling-failure probability. Let $Z=(Z_1,\ldots,Z_m)$ be binary random variables and let $\Pi$ be a uniformly random $k$-subset of $[m]$, independent of $Z$. The sampling-failure event then has probability
\begin{equation}
 \Pr\left[\frac1k\sum_{i\in\Pi}Z_i<\Qtol\ \text{and}\ \frac1n\sum_{i\notin\Pi}Z_i\geq\Qtol+\nu\right].
\end{equation}
In our security analysis, the sampling-failure probability $\eps_\pe(\nu)$ contributes to $\eps_{\mathrm{tot}}$, so a larger $\eps_\pe(\nu)$ results in a shorter secure key. Tighter estimates of this probability therefore let our proof certify a longer key from the same data.

Previous analyses, e.g.~\cite{tl2017}, use the Serfling bound
\begin{align}
    \Pr\left[\frac1k\sum_{i\in\Pi}Z_i<\Qtol\ \text{and}\ \frac1n\sum_{i\notin\Pi}Z_i\geq\Qtol+\nu\right] \leq \exp\!\left[-\frac{2nk^2\nu^2}{m(k+1)}\right].
    \label{eq:serfling}
\end{align}

Following~\cite{mannalath2025}, we evaluate the tails exactly. For a uniformly chosen $k$-subset $\Pi$ of the $m$ rounds, independent of the binary pattern $z$, define
\begin{equation}
 \eps_\pe(\nu):=\max_{z\in\{0,1\}^m}
 \Pr_\Pi\!\left[
 \frac1k\sum_{i\in\Pi}z_i<\Qtol\ \text{and}\
 \frac1n\sum_{i\notin\Pi}z_i\geq\Qtol+\nu\right].
 \label{eq:peprob}
\end{equation}

For a fixed total number of errors, the probability being maximized is a hypergeometric tail, since $\Pi$ samples rounds without replacement. Thus, to compute $\eps_\pe(\nu)$ it suffices to maximize over these hypergeometric CDFs. This can be done efficiently because the maximizer can be solved for; see Appendix~\ref{app:sampling}.

Another statistical question arises in choosing the tolerated quantum bit error rate (QBER) $\Qtol$ such that the protocol satisfies the desired completeness parameter. For independent physical errors of probability $Q$, the number of errors in a test sample of size $k$ is binomial, so the probability that the test aborts is exactly
\begin{equation}
 \Pr[\operatorname{Binom}(k,Q)\geq k\Qtol].
 \label{eq:test-abort}
\end{equation}
We invert this tail to choose the smallest threshold on the implemented integer grid meeting the assigned test-abort budget. Evaluating the distribution directly avoids the excess tolerance introduced by an analytical concentration bound. This matters for key extraction because a larger $\Qtol$ also weakens the entropy estimate. 

We compare the gains from evaluating the tails exactly rather than using concentration bounds. Figure~\ref{fig:sampling-comparison} separates the gains from tail evaluation and hashing. The baseline is the parameter estimation analysis of~\cite{tl2017}, and the comparison quantifies improvements to that analysis. The curves labelled `exact tails' use binomial completeness and hypergeometric parameter-estimation probabilities. Those labelled `analytic tails' use the Serfling bound
in \ref{eq:serfling} which bounds the event with non-strict test acceptance and hence also the strict event here. For the completeness threshold, the analytic-tail curves use the Chernoff bound~\cite[Appendix~III]{mannalath2025}
\begin{equation}
 \Pr[\operatorname{Binom}(k,Q)\geq k\Qtol]
 \leq \exp[-k\,d(\Qtol\|Q)],\qquad Q<\Qtol\leq1,
 \label{eq:chernoff-completeness}
\end{equation}
where $d(p\|q)=p\ln(p/q)+(1-p)\ln[(1-p)/(1-q)]$ uses natural logarithms. We choose the smallest positive integer multiple of $1/k$ for $\Qtol$ that makes this upper bound at most $\eps_{\mathrm{com}}/2$. Each choice of tails is combined with both LHLs, so the comparison includes their effect on both the tolerated error rate and the entropy estimate.

\FloatBarrier

\subsection{Numerical model and implementation}\label{sec:methods-numerics}

The Python code and numerical data used to generate the figures are available on Zenodo~\cite{leftoverscode2026}. For our numerics, we compute the $\eps_{\rm tot}$-secure and $\eps_{\rm com}$-complete key rates guaranteed by the EUR-based security proof we consider and the entropy accumulation-based security proof in \cite{gohtan2026}, at various block sizes $m$ and QBER $Q$. We assume independent errors for this performance model, with the same $Q$ in the two BB84 bases. This can be modeled as Alice and Bob sharing a suitably depolarized Bell state. 

The BB84 protocol we consider can abort at two stages, parameter estimation and error correction. We choose protocol parameters such that the probability of aborting at each of these stages is bounded by $\frac{\eps_{\rm com}}{2}$, and thus the total probability of aborting is bounded by $\eps_{\rm com}$. 

Since each round of communication has probability $Q$ of error, the number of errors in the test sample is given by $\operatorname{Binom}(k,Q)$. To satisfy the completeness condition for the parameter estimation stage, we choose a tolerated QBER $\Qtol$, such that 
\begin{equation}
    \Pr[\operatorname{Binom}(k,Q)\geq k\Qtol] \leq \frac{\eps_{\rm com}}{2}.
\end{equation}

For the one-way error correction stage of the protocol, we assume the leakage model of Ref.~\cite{reconciliation2017}, based on a finite-block normal approximation fitted to low-density parity-check (LDPC) code simulations:
\begin{equation}
 \leakEC\simeq\xi_1nh(Q)-\xi_2\sqrt{nV(Q)}\,
 \Phi^{-1}(\eps_{\mathrm{LDPC}}),\qquad
 V(Q)=Q(1-Q)\left[\log_2\frac{1-Q}{Q}\right]^2.
 \label{eq:leakage}
\end{equation}
Here $\Phi$ is the standard normal cumulative distribution function. The term proportional to $n$ models the leading reconciliation cost, while the square-root term accounts for finite block length and the requested decoding reliability. Thus to achieve the desired probability of aborting we choose $\eps_{\rm LDPC} =\frac{\eps_{\mathrm{com}}}{2}$. For Fig.~\ref{fig:rates}, we use $(\xi_1,\xi_2)=(1.13,3.82)$ at $Q=1\%$ and $(1.07,3.71)$ at $Q=2.5\%$, the respective fits in Ref.~\cite[Eq.~(21), Table~I]{reconciliation2017} at decoding-failure probability $1\%$ and reconciliation block lengths $10^3$ to $10^6$. Figure~\ref{fig:sampling-comparison} also uses $(1.13,3.82)$ at $Q=1\%$. For both panels of Fig.~\ref{fig:satellite}, we assume $(1.13,3.82)$ throughout the QBER range, extending the $1\%$ fit to other $Q$. Using these fits outside their fitted block lengths and at decoding-failure probability $\eps_{\mathrm{com}}/2$ is an extrapolation. Attaining the resulting leakage and reliability is a modeling assumption. 

The asymptotic limit $1-2h(Q)$ shown in Fig.~\ref{fig:rates} assumes capacity-achieving error correction. Our fitted coefficients instead leave an asymptotic leakage of $\xi_1h(Q)$ bits per key round, so the model approaches $1-(1+\xi_1)h(Q)$ when the statistical overhead and test fraction vanish. It therefore cannot reach the plotted limit for $Q>0\%$: extrapolating this finite-block fit to arbitrarily large $n$ is too pessimistic to describe codes approaching capacity.

Since $t$ enters the key length
only through $2^{-t}$ and $H_{*}(\nu,t)$, and $\ell$ is concave in
$2^{-t}$, we can solve analytically for the optimal $t$. This is
equivalent to fixing the share of the security budget $\varepsilon_{\rm tot}$ allotted to
the error correction phase, since $\varepsilon_{\rm cor}(t) = 2^{-t}$. We obtain $\varepsilon_{\rm cor}(t) = \frac{1}{3}(\varepsilon_{\rm tot} - \varepsilon_{\rm pe}(\nu))$. We optimize over
the remaining parameters $k$ and $\nu$ numerically by a grid search (see Appendix \ref{app:numerics} for details).

For the key rates computed for the EAT-based proof, we use the analysis given by \cite{gohtan2026} with modifications to incorporate the completeness condition. The BB84 protocol they consider varies from ours in that each round can be either a generation round or a test round, with some specified probability.Here $\gamma$ is the probability that a round is declared a test round, and
$\delta > 0$ is a protocol parameter: the protocol aborts if fewer than
$m(1-\gamma-\delta)$ generation rounds occur. Thus the number of generation rounds is itself random, distributed as $\operatorname{Binom}(m,1-\gamma)$, and the number of errors observed in test rounds is distributed as $\operatorname{Binom}(m,\gamma Q)$. The completeness analysis therefore also has to account for fluctuations in the number of generation rounds. As in the EUR-based case, we allocate $\frac{\eps_{\rm com}}{2}$ to error correction. The remaining $\frac{\eps_{\rm com}}{2}$ is split equally, $\frac{\eps_{\rm com}}{6}$ each, between the number of generation rounds exceeding $n_{\mathrm{eff}}$, falling below $m(1-\gamma-\delta)$, and parameter estimation (where this condition is added to condition for the parameter estimation stage of the protocol to accept in \cite{gohtan2026}). For this protocol, the parameters of the error correction code need to be decided a priori, but the message length, which is the number of generation rounds, is a random variable. For the sake of our comparison, we assume that the error correction code is designed with a number of generation rounds $n_{\rm eff}$ in mind, and thus can correct errors for $n_{\rm gen}\leq n_{\rm eff}$. The leakage is then computed from \eqref{eq:leakage} with $n = n_{\mathrm{eff}}$ and $\eps_{\mathrm{LDPC}} = \frac{\eps_{\rm com}}{2}$.

\vspace*{-.2\baselineskip}\enlargethispage{\baselineskip}
\begin{acknowledgments}
\vspace*{-.6\baselineskip}

We are indebted to Ernest Y.-Z.\ Tan for discussions on numerical simulations and comments on a draft of the paper, and to Jun Hui Goh for explaining the code used to compute the key rates in~\cite{gohtan2026}.

M.S.T.\ acknowledges support by the Dieter Schwarz Foundation, and from the National Quantum Scholarships Scheme (Master) Scholarship.
B.R.\ acknowledges the support of the Japan Science and Technology Agency (JST) PRESTO grant no.\ JPMJPR25FB. M.T.\ is supported by the National Research Foundation, Prime Minister's Office, Singapore under its Campus for Research Excellence and Technological Enterprise programme, via the project Quantum Security and Resilience for Emerging Technologies (QUASAR-CREATE).

\smallskip
\textbf{Statement on use of AI.}\,---\,ChatGPT 6 Astra and Claude Opus 5 were used to assist with the numerical simulations, refining the mathematical proofs, and with the preparation of the manuscript. The authors take responsibility for any errors.

\end{acknowledgments}
\clearpage


\begin{table}[!htbp]
 \centering
 \small
 \begin{tabular}{@{}lr@{\hspace{10pt}}r@{}}
  \toprule
  Parameter & Figure~\ref{fig:rates} & Figure~\ref{fig:satellite}(a) \\
  \midrule
  Total block size $m$ & $100000$ & $2000$ \\
  Test and key blocks $(k,n)$ & $(6495,93505)$ & $(561,1439)$ \\
  QBER $Q$ & $1\%$ & $1\%$ \\
  Test threshold $\Qtol$ & $1.339\%$ & $2.317\%$ \\
  Total security target $\eps_{\mathrm{tot}}$ & $10^{-10}$ & $10^{-10}$ \\
  Completeness target $\eps_{\mathrm{com}}$ & $1\%$ & $1\%$ \\
  Test and reconciliation abort budgets & $0.5\%\text{ each}$ & $0.5\%\text{ each}$ \\
  Leakage coefficients $(\xi_1,\xi_2)$ & $(1.13,3.82)$ & $(1.13,3.82)$ \\
  Analysis parameter $\nu$ & $1.151\%$ & $7.575\%$ \\
  Syndrome length $\leakEC$ (bits) & $10522$ & $378$ \\
  Verification length $t$ (bits) & $39$ & $37$ \\
  Privacy amplification coefficient $\kappa$ & $4/3$ & $4/3$ \\
  Sampling failure bound $\eps_\pe(\nu)$ & $9.430\times10^{-11}$ & $7.520\times10^{-11}$ \\
  Correctness bound $\eps_{\mathrm{cor}}$ & $1.819\times10^{-12}$ & $7.276\times10^{-12}$ \\
  PA term in Eq.~\eqref{eq:numerical-new-security} & $3.344\times10^{-12}$ & $1.648\times10^{-11}$ \\
  Total security bound & $9.946\times10^{-11}$ & $9.896\times10^{-11}$ \\
  Secret key length $\ell$ (bits) & $67147$ & $284$ \\
  Key rate $\ell/m$ & $0.671$ & $0.142$ \\
  \bottomrule
 \end{tabular}
 \caption{\textbf{Parameters of the marked new LHL points in Figures~\ref{fig:rates} and~\ref{fig:satellite}(a).} Both points have $Q=1\%$ and use exact hypergeometric sampling, the reconciliation model and the privacy amplification bound in equation~\eqref{eq:numerical-new-security}. Displayed decimal values are rounded.}
 \label{tab:example-parameters}
\end{table}

\bibliographystyle{apsca}
\bibliography{references}

\FloatBarrier
\clearpage
\appendix
\section{Leftover hashing}\label{app:proofs}

In this appendix we discuss the technical details of the privacy amplification bounds in terms of smooth entropies, in particular giving a complete proof of Theorem~\ref{thm:pa} together with its extensions.

\subsection{Entropy, distance and hashing conventions}\label{app:conventions}

We use the following entropy and hashing conventions throughout, which are consistent with~\cite{tl2017}. The general background is developed in~\cite{tomamichelbook}.

All logarithms are to base two. 
For $\rho_{AB}\geq0$ with $\Tr\rho_{AB}\leq1$, we define the min-entropy as
\begin{equation}\begin{aligned}
 \hmin(A|B)_\rho &= - \inf_{\sigma_B} D_{\max}(\rho_{AB} \| \id_A \otimes \sigma_E)\\
 &=  \sup\{\lambda:\rho_{AB}\leq2^{-\lambda}\id_A\otimes\sigma_B,
                    \ \sigma_B\geq0,\ \Tr\sigma_B=1\},
\end{aligned}\end{equation}
which is sometimes also denoted $H_{\min}^{\uparrow}(A|B)$ in the literature due to the involved optimisation over $\sigma_B$. 
The generalised root fidelity and purified distance are
\begin{align}
 \overline F(\rho,\rho')&=
 \norm{\sqrt\rho\sqrt{\rho'}}_1+
 \sqrt{(1-\Tr\rho)(1-\Tr\rho')},\\
 P(\rho,\rho')&=\sqrt{1-\overline F(\rho,\rho')^2},
\end{align}
where $\norm{X}_1 = \Tr\sqrt{XX^\dagger}$ is the trace norm. 
The smoothed quantity $\hmin^\eps$ is the supremum of the min-entropy over subnormalised states $\rho'$ that are $\eps$-close to $\rho$:
\begin{equation}\begin{aligned}
    \hmin^\eps(A|B)_\rho &= \sup \left\{ \hmin(A|B)_{\rho'} : \rho'_{AB} \geq 0,\; \Tr \rho'_{AB} \leq 1,\; P(\rho_{AB},\rho'_{AB}) \leq \eps \right\},
\end{aligned}\end{equation}
where $\eps<\sqrt{\Tr\rho}$.

Let $\mathcal H=\{h_s:\mathcal X\to\mathcal Z\}_{s\in\mathcal S}$ be a finite indexed hash family. A seed $S$, uniform on $\mathcal S$ and independent of $XE$, selects $h=h_S$; the output is $Z=h(X)$. Let $\pi_Z=\id_Z/|\mathcal Z|$ and write
$\omega_{ZE}^{s}=\sum_x|h_s(x)\rangle\langle h_s(x)|\otimes\rho_E^x$
for the output when $\rho_{XE}=\sum_x|x\rangle\langle x|\otimes\rho_E^x$ and $S=s$. Because $S$ is classical and independent of $XE$, the secrecy quantity~\eqref{eq:secrecy} is also the trace distance between $\omega_{ZSE}$ and $\pi_Z\otimes\omega_{SE}$, so it explicitly includes the publicly revealed seed. We call a family $2^*$-universal when
\begin{equation}
 \Pr_S[h_S(x)=h_S(x')]=\frac1{|\mathcal Z|}\qquad(x\ne x').
 \label{eq:exactcollision}
\end{equation}
This is the exact-collision convention found for example in Refs.~\cite{tl2017,dupuis2023,rt2026}. It is a stronger assumption than ordinary 2-universality (which would replace equality in equation~\eqref{eq:exactcollision} with an upper bound) but a weaker one than strong 2-universality, which would require pairwise independent uniform outputs~\cite{wegman_1981}. In Section~\ref{app:hash} we will show how our main result can be extended to 2-universal and almost 2-universal hash families, at the expense of incurring an additional error term.

\subsection{Measured smoothing and its operator representation}\label{app:measured}

We denote the measured smooth collision divergence by $D_2^{\eta,\mathbb{M}}$ following Ref.~\cite{rt2026}. To formalise its definition, we begin with the classical (commuting) case. For a subnormalised classical distribution $p$ and a nonnegative reference vector $q$, define
\begin{equation}
 D_2^{\eta,\mathrm T}(p\|q)
 \coloneqq\log_2\inf_{\substack{0\leq p'\leq p\\\sum_x(p_x-p'_x)\leq\eta}}
                  \sum_x\frac{(p'_x)^2}{q_x}.
 \label{eq:classicald2}
\end{equation}
The superscript T here indicates that this is equivalent to smoothing over a generalised total variation distance (trace distance) ball: taking the componentwise minimum of any feasible distribution with $p$ cannot increase the objective or its distance from $p$, so the restriction to $p' \leq p$ is done without loss of generality. For consistency we take $0/0=0$ and $a/0=+\infty$ for $a>0$. For a subnormalised state $\rho$, a positive operator $\sigma$ and $0\leq\eta<\Tr\rho$, the quantum definition is then
\begin{equation}
 D_2^{\eta,\mathbb{M}}(\rho\|\sigma)
 \coloneqq\sup_{\mathcal M\in\mathbb M}
 D_2^{\eta,\mathrm T}\bigl(\mathcal M(\rho)\|\mathcal M(\sigma)\bigr),
 \label{eq:measuredd2}
\end{equation}
where $\mathbb M$ contains all quantum-to-classical measurement channels~\cite[Definition~1]{rt2026}.

Although this definition involves an optimisation over measurements, its useful feature for the proof is an equivalent description directly in terms of operators. Theorem~2 of~\cite{rt2026} shows that the measured smoothing can be represented by a Hermitian approximation $R\leq\rho$, with $\Tr(\rho-R)\leq\eta$. Here $R$ may have negative eigenvalues: it is an auxiliary object in the proof, not to be interpreted as a physical state. The residual $\rho-R$ remains positive and has small trace, so its contribution to the final security error can be easily controlled. The quadratic cost of $R$ is then measured by the measured R\'enyi relative entropy of order 2, which corresponds to the Bures norm $\norm{R}_\sigma^2$. For $\sigma>0$, this norm is
\begin{equation}
 \norm{R}_\sigma^2=\Tr[R\mathcal J_\sigma^{-1}(R)],
 \qquad \mathcal J_\sigma(B)=\tfrac12(\sigma B+B\sigma).
 \label{eq:bures}
\end{equation}
For diagonal operators this reduces to the weighted sum $\sum_x R_x^2/\sigma_x$ appearing in the classical definition. The variational representation of~\cite[Theorem~2]{rt2026} is then
\begin{equation}
 2^{D_2^{\eta,\mathbb{M}}(\rho\|\sigma)}
 =\inf_{\substack{R=R^\dagger,\ R\leq\rho\\\Tr(\rho-R)\leq\eta}}
       \norm R_\sigma^2.
 \label{eq:hermitiand2}
\end{equation}
The same definition and variational form can be extended to subnormalised inputs --- although stated there for quantum states, it is explicit that the proofs of~\cite{rt2026} do not make use of the normalisation of $\rho$. For singular $\sigma$, the definition is extended to the (semi)norm defined by $\norm R_\sigma^2=\lim_{\tau\downarrow0}\norm R_{\sigma+\tau\id}^2$, possibly infinite.

\subsection{Sharper comparison with smooth min-entropy}\label{app:comparison}

\begin{proposition}[Comparison with smooth min-entropy]\label{prop:entropycomparison}
Let $\rho_{XE}$ be a subnormalised classical--quantum state and $0\leq\eps<\sqrt{\Tr\rho}$. Then
\begin{equation}
 \inf_{\sigma_E\in\mathcal S(E)}
 D_2^{\eps^2,\,\mathbb{M}}(\rho_{XE}\|\id_X\otimes\sigma_E)
 \leq-\hmin^\eps(X|E)_\rho+\log_2\kappa(\eps),
 \label{eq:appendixentropycomparison}
\end{equation}
where $\mathcal S(E)$ denotes the normalised states on $E$ and $\kappa$ is given by
\begin{equation}
 \kappa(\eps)=
 \begin{cases}
  (1-\eps^2)(1+3\eps^2)&\text{ if } 0\leq\eps\leq1/\sqrt3\\
  4/3&\text{ if } 1/\sqrt3<\eps<1
 \end{cases} \,\leq \,\frac43\,.
 \label{eq:kappa_appendix}
\end{equation}
In particular, the additive correction tends to zero as $\eps\to0$.
\end{proposition}

We stress here that the smoothing parameter on the left is $\eps^2$ in the measured smoothing sense of~\cite{rt2026}, whereas $\hmin^\eps$ uses the conventional purified distance smoothing. This sharper comparison is the new ingredient here that improves on the estimates found in~\cite[Corollary~15]{rt2026}. 

We will use a variational form of the Bures norm, which can already be deduced from the expressions derived in~\cite{rt2026}, and which we prove here for completeness.
\begin{lemma}\label{lem:variational}
For any $R = R^\dagger$ and $\sigma \geq 0$, the Bures seminorm can be expressed as
\begin{equation}\begin{aligned}
  \norm{R}_{\sigma} &= \inf \left\{ \norm{Y}_{2} \;:\; R = \Ree\!\left(  Y \sqrt{\sigma} \right) \right\}
\end{aligned}\end{equation}
where $\Ree(Z) = \frac12(Z+Z^\dagger)$ denotes the Hermitian part of a matrix.
\end{lemma}
\begin{proof}
We begin similarly as in Corollary~7 of~\cite{rt2026}, establishing a dual expression
\begin{equation}\begin{aligned}
  \norm{R}_{\sigma}^2 &\texteq{(i)} \sup_{H=H^\dagger} \left[ 2 \Tr H R - \Tr H^2 \sigma \right]\\[-3pt]
  &\texteq{(ii)} \sup \lset 2 \Tr H R - \Tr C \sigma \bar  C \geq H^2 \rset\\
  &\texteq{(iii)} \sup \lset 2 \Tr H R - \Tr C \sigma \bar  \begin{pmatrix}C & H \\ H & \id \end{pmatrix} \geq 0 \rset\\
  &\texteq{(iv)} \inf \lset \Tr Q_{22} \bar Q = \begin{pmatrix} Q_{11} & Q_{12} \\ Q_{21} & Q_{22} \end{pmatrix} \geq 0,\; Q_{12} + Q_{21} = 2 R,\; Q_{11} = \sigma \rset.
\end{aligned}\end{equation}
Here, (i) is a variational form shown e.g.\ in the proof of~\cite[Lemma~4]{rt2026}, (ii) follows since the positivity of $\sigma$ ensures that $\Tr C \sigma \geq \Tr H^2 \sigma$ for any $C \geq H^2$, (iii) is a standard Schur complement result~\cite[Thm.~1.3.3, Ex.~1.3.5]{bhatia_2007}, and (iv) is by strong Lagrange duality. Using Schur complements again, the positivity of $Q$ is equivalent to the conditions that $Q_{22} \geq 0$, $Q_{21} = Q_{12}^\dagger$ and $Q_{22} \geq Q_{12}^\dagger Q_{11}^{-1} Q_{12}$, and hence 
\begin{equation}\begin{aligned}\label{eq:variational_quadratic}
  \norm{R}_{\sigma}^2 &= \inf \lset \Tr Q_{22} \bar Q_{22} \geq Q_{12}^\dagger \sigma^{-1} Q_{12},\; Q^\dagger_{12} + Q^{\vphantom{\dagger}}_{12} = 2 R \rset\\
  &= \inf \lset \Tr Q_{12}^\dagger \sigma^{-1} Q_{12} \bar Q^\dagger_{12} + Q^{\vphantom{\dagger}}_{12}= 2 R \rset\\
  &= \inf \lset \Tr Y^\dagger Y \bar  Y\sigma^{1/2} + \sigma^{1/2}Y^\dagger = 2 R \rset
\end{aligned}\end{equation}
where we defined $Y = Q^\dagger_{12} \sigma^{-1/2} $.
\end{proof}

The result then follows by constructing an explicit feasible solution for the variational form of $D_2^{\eps^2,\mathbb{M}}$ from a feasible approximation of $\rho$ in purified distance. We state the exact estimate separately for clarity.

\begin{lemma}[Hermitian approximation of purified distance]
\label{lem:construction}
Let $\rho,\rho'$ be subnormalised states with $P(\rho,\rho') = \sqrt{1 - g^2}$ where $g=\overline F(\rho,\rho')$ is the generalised root fidelity. Let $\sigma$ be any state such that $\rho'\leq\lambda\sigma$. Then there exists a Hermitian operator $R\leq\rho$ such that
\begin{equation}
 \norm{\rho-R}_{1} = \Tr(\rho-R)\leq1-g^2,
 \qquad \norm R_\sigma^2\leq\lambda g^2(4-3g^2).
\end{equation}
If $\rho$ is classical--quantum, $R$ can be taken to be classical--quantum too.
\end{lemma}
\begin{proof}
Our construction is based on the definition of Bures distance, which we recall to be
\begin{equation}\begin{aligned}
  \min_{U \text{ unitary}} \norm{\sqrt{\rho} - U \sqrt{\rho'}}_{2}^2 &= \min_{U \text{ unitary}} \Tr \left( \sqrt{\rho} - U \sqrt{\rho'} \right)^\dagger \left( \sqrt{\rho} - U \sqrt{\rho'} \right)\\
  &= \Tr \rho + \Tr \rho' - \max_{U \text{ unitary}} 2 \Ree \Tr U \sqrt{\rho'} \sqrt{\rho}\\
  &= \Tr \rho + \Tr \rho' - 2 f,
\end{aligned}\end{equation}
where $f=\norm{\sqrt\rho\sqrt{\rho'}}_1$. We will take $U$ to be the optimal unitary above, so that $\Ree \Tr U \sqrt{\rho'} \sqrt{\rho} = f$. Notice now that if we take the operator $\sqrt{\rho} - U \sqrt{\rho'}$ and rescale the second term by a factor of $g = f + \sqrt{(1-\Tr \rho)(1- \Tr \rho')}$,  
then the squared operator
\begin{equation}\begin{aligned}
  Q &\coloneqq \left( \sqrt{\rho} - g\, U \sqrt{\rho'} \right)^\dagger \left( \sqrt{\rho} - g \,U \sqrt{\rho'} \right) = \rho + g^2 \rho' - 2 g \Ree\!\left(\sqrt{\rho} U \sqrt{\rho'} \right)
\end{aligned}\end{equation}
satisfies
\begin{equation}\begin{aligned}
  \Tr Q &= \Tr \rho + g^2 \Tr \rho' - 2 f g\\
  &= 1 - g^2 - \left(\sqrt{1-\Tr\rho} - g \sqrt{1-\Tr\rho'}\right)^2\\
  &\leq 1 - g^2.
\end{aligned}\end{equation}
Let us then take
\begin{equation}\begin{aligned}
  R \coloneqq \rho - Q = \Ree\!\left( 2 g \sqrt{\rho} U \sqrt{\rho'} - g^2 \rho'\right),
\end{aligned}\end{equation}
which satisfies $R \leq \rho$ since $Q \geq 0$ by definition. 

Let now $L \coloneqq \sqrt{\rho'}\sigma^{-1/2}$, with the inverse taken on the support of $\sigma$. By assumption, this operator satisfies
\begin{equation}\begin{aligned}
  \norm{L}_{\infty}^2 = \norm{\sigma^{-1/2} \rho' \sigma^{-1/2}}_{\infty} =  2^{D_{\max}(\rho'\|\sigma)} \leq \lambda.
\end{aligned}\end{equation}
Furthermore, 
\begin{equation}
 R=\Ree\!\left[(2g\sqrt\rho U-g^2\sqrt{\rho'})L\sqrt\sigma\right],
\end{equation}
which we observe to form a feasible solution for the variational expression in Lemma~\ref{lem:variational}.
Thus
\begin{align}
 \norm R_\sigma^2 &\leq \norm{\left(2g \sqrt\rho U - g^2 \sqrt{\rho'} \right) L\,}_{2}^2 \nonumber\\
 &\leq\lambda\norm{2g\sqrt\rho U-g^2\sqrt{\rho'}}_2^2 \nonumber\\
 &=\lambda\left(4g^2\Tr\rho-4g^3f+g^4\Tr\rho'\right)\\
 &=\lambda \left( 4 g^2 - 4 g^4 + g^4 - g^2 \left[ 2 \sqrt{1-\Tr\rho} - g \sqrt{1-\Tr\rho'} \right]^2 \right) \nonumber\\
 &\leq\lambda g^2(4-3g^2). \nonumber
\end{align}
For a classical--quantum $\rho$, choosing blockwise polar unitaries in the construction also makes $R$ classical--quantum, as needed for hashing.
\end{proof}

\begin{proof}[Proof of Proposition~\ref{prop:entropycomparison}]
Choose a state $\rho'_{XE}$ with $P(\rho,\rho')\leq\eps$ and a normalised state $\sigma_E$ such that $\rho'_{XE}\leq\lambda\id_X\otimes\sigma_E$. For simplicity and without loss of generality, we can restrict ourselves to full-rank $\sigma_E$: any singular $\sigma_E$ can be approximated by a full-rank state at the expense of increasing $\lambda$ slightly, and the result then follows by a limiting argument. Now, we may dephase $\rho'$ on $X$: this preserves the operator inequality and cannot increase purified distance from the classical--quantum $\rho$ due to the data processing inequality. Apply Lemma~\ref{lem:construction} with $\sigma=\id_X\otimes\sigma_E$. Since $g^2\in[1-\eps^2,1]$, maximising $x(4-3x)$ on this interval gives $\kappa(\eps)$: the maximum is attained at $x=1-\eps^2$ if $\eps^2\leq1/3$, and at $x=2/3$ otherwise. Thus the resulting $R$ is feasible in equation~\eqref{eq:hermitiand2} with $\eta=\eps^2$, and
\begin{equation}
 D_2^{\eps^2,\,\mathbb{M}}(\rho_{XE}\|\id_X\otimes\sigma_E)
 \leq\log_2\lambda+\log_2\kappa(\eps).
\end{equation}
Taking the infimum over all $\rho'_{XE}$ and all $\sigma_E$ yields equation~\eqref{eq:appendixentropycomparison} by the definition of $\hmin^\eps$.
\end{proof}

\subsection{Proof of Theorem~\ref{thm:pa}}\label{subsec:theorem1}

We recall the full statement of the theorem.

{
\renewcommand{\thetheorem}{\ref{thm:pa}}
\begin{theorem}
Let $\rho_{XE}$ be a finite-dimensional subnormalised classical--quantum state and let $0\leq\eps<\sqrt{\Tr\rho_{XE}}$. For any $2^*$-universal family of hash functions,
\begin{equation}
 \Delta(\rho,\mathcal H)
 \leq \eps^2+\frac12\sqrt{\big(|\mathcal Z|-1\big)\,\kappa(\eps)\,
             2^{-\hmin^\eps(X|E)_\rho}},
\end{equation}
where $\kappa(\eps)$ is as in~\eqref{eq:kappa_appendix}.
\end{theorem}
}

The proof follows directly by combining the measured smooth leftover hash lemma of~\cite[Theorem~14]{rt2026} with the improved purified distance bound in Proposition~\ref{prop:entropycomparison}. For completeness, below we go through the hashing argument of~\cite[Lemma~12 and Theorem~14]{rt2026} in detail, which will help us with generalising the result beyond 2*-universal hashes.

\begin{proof}
Recall that we are interested in bounding
\begin{equation}
 \Delta(\rho,\mathcal H) = \mathbb E_S\left[\frac12\norm{\omega_{ZE}^{S}-\pi_Z\otimes\rho_E}_1\right]
 \label{eq:secrecy_app}
\end{equation}
for a given hash family $\mathcal{H}$ with seed $S$, where $\omega_{ZE}^{S}=\sum_x|h_S(x)\rangle\langle h_S(x)|\otimes \rho_x$. 
Let $R_{XE}=\sum_x|x\rangle\langle x|\otimes R_x$ be a Hermitian operator such that $R_{XE} \leq \rho_{XE}$ and $\sigma_E$ be a normalised state, assumed for simplicity to be strictly positive. Define
\begin{equation}
 R_{ZE}^{S}=\sum_x|h_S(x)\rangle\langle h_S(x)|\otimes R_x,
 \qquad R_E=\sum_xR_x.
\end{equation}
For any Hermitian $A$ and positive $\tau$, the Cauchy--Schwarz inequality can be used to bound~\cite{temme_2010,rt2026}
\begin{equation}
 \norm A_1 \leq\sqrt{\Tr\tau}\,\norm A_\tau,
\end{equation}
so taking $\tau=\id_Z\otimes\sigma_E$ and averaging over $S$ gives
\begin{equation}\begin{aligned}
 \mathbb E_S\frac12\norm{R_{ZE}^{S}-\pi_Z\otimes R_E}_1
 &\leq\frac12 \sqrt{|\mathcal Z|} \, \mathbb E_S 
       \norm{R_{ZE}^{S}-\pi_Z\otimes R_E}_{\id_Z\otimes\sigma_E}\\
 &\leq\frac12 \sqrt{|\mathcal Z| \, \mathbb E_S 
       \norm{R_{ZE}^{S}-\pi_Z\otimes R_E}_{\id_Z\otimes\sigma_E}^2}
\end{aligned}\end{equation}
by Jensen's inequality. The reason why we take the expectation of the squared term here is that it allows us to expand the expression as 
\begin{align}
 \mathbb E_S\norm{R_{ZE}^{S}-\pi_Z\otimes R_E}_{\id_Z\otimes\sigma_E}^2
 &=\sum_{x,x'}G_{xx'}\langle R_x,R_{x'}\rangle_{\sigma_E}
 \label{eq:collisionexpand}
\end{align}
using the Bures inner product $\langle A,B\rangle_{\sigma_E}=\Tr[A\mathcal J_{\sigma_E}^{-1}(B)]$ and the  real symmetric collision matrix $G$ defined as
\begin{equation}\label{eq:Gmatrix}
 G_{xx'} \coloneqq \Pr_S[h_S(x)=h_S(x')]-\frac1{|\mathcal Z|}.
\end{equation}
We then combine the above with the triangle inequality and the data processing inequality for trace distance to bound
\begin{equation}\begin{aligned}
 \Delta(\rho,\mathcal H) &=\frac12 \mathbb E_S\norm{\omega_{ZE}^{S}-\pi_Z\otimes\rho_E}_1\\
 &\leq \frac12 \mathbb E_S \norm{ \omega_{ZE}^S - R_{ZE}^S}_1 + \frac12 \mathbb E_S \norm{ \pi_Z\otimes \rho_E - \pi_Z\otimes R_E}_1 + \frac12 \mathbb E_S \norm{ R_{ZE}^S - \pi_Z\otimes R_E}_1  \\
 &\leq \mathbb E_S \norm{ \omega_{ZE}^S - R_{ZE}^S}_1 + \frac12 \mathbb E_S \norm{ R_{ZE}^S - \pi_Z\otimes R_E}_1  \\
 &\leq \norm{ \rho_{XE} - R_{XE}}_1 + \frac12 \sqrt{|\mathcal Z|\, \mathbb E_S\norm{R_{ZE}^{S}-\pi_Z\otimes R_E}_{\id_Z\otimes\sigma_E}^2}  \\
 &= \Tr(\rho-R) + \frac12 \sqrt{|\mathcal Z| \,\sum_{x,x'}G_{xx'}\langle R_x,R_{x'}\rangle_{\sigma_E}}.
 \label{eq:generalhash}
\end{aligned}\end{equation}
In particular, under the assumption of 2*-universality (equation~\eqref{eq:exactcollision}) all off-diagonal entries of $G$ vanish, so
\begin{equation}\begin{aligned}
 \sum_{x,x'}G_{xx'}\langle R_x,R_{x'}\rangle_{\sigma_E} &= \left(1-\frac{1}{|\mathcal Z|}\right) \sum_x \norm{R_x}_{\sigma_E}^2\\
 &= \left(1-\frac{1}{|\mathcal Z|}\right) \norm{R_{XE}}_{I_X \otimes \sigma_E}^2,
\end{aligned}\end{equation}
which gives
\begin{equation}
 \Delta(\rho,\mathcal H)\leq\Tr(\rho-R)
 +\frac12\sqrt{|\mathcal Z|-1}\,\norm R_{\id_X\otimes\sigma_E}.
 \label{eq:hermitianlhl}
\end{equation}
This expression, upon a minimisation over $R$ and $\sigma_E$, gives the statement of the measured smooth leftover hash lemma of~\cite{rt2026}. Our improvement over the relaxed bound of~\cite[Corollary~15]{rt2026} then follows as a result of Proposition~\ref{prop:entropycomparison}. 
More precisely, choose a classical--quantum smoothing state $\rho'_{XE}$ with $P(\rho,\rho')\leq\eps$ and $\rho'_{XE}\leq\lambda\id_X\otimes\sigma_E$.  Lemma~\ref{lem:construction} then provides a classical--quantum operator $R\leq\rho$ satisfying
\begin{equation}
 \Tr(\rho-R)\leq\eps^2,\qquad
 \norm R_{\id_X\otimes\sigma_E}^2\leq\lambda\,\kappa(\eps).
 \label{eq:almostapprox}
\end{equation}
Recalling that the infimum of $\lambda$ over all such $\rho'_{XE}$ and $\sigma_E$ is precisely $2^{-\hmin^\eps(X|E)}$, plug this into Eq.~\eqref{eq:hermitianlhl} and voil\`a.
\end{proof}

\subsection{Almost 2-universal families}\label{app:hash}

We now extend the common argument beyond exact collisions. For $\mu\geq0$, assume
\begin{equation}
 \Pr_S[h_S(x)=h_S(x')]\leq\frac{1+\mu}{|\mathcal Z|}
 \qquad(x\ne x').
 \label{eq:almostcollision}
\end{equation}
The case $\mu=0$ is ordinary 2-universality.

\begin{proposition}[Almost 2-universal hashing]
\label{prop:ordinary}
Let $\rho_{XE}$ be a finite-dimensional subnormalised classical--quantum state and let $\eps \in [0, \sqrt{\Tr\rho_{XE}})$. For any family of hash functions satisfying~\eqref{eq:almostcollision}, we have
\begin{equation}\begin{aligned}
 \Delta(\rho,\mathcal H) &\leq \eps^2 +\frac12\sqrt{
 2\bigl[|\mathcal Z|-1+\mu(|\mathcal X|-1)\bigr]
 \,2^{\inf_{\sigma_E\in\mathcal S(E)} D^{\eps^2,\mathbb{M}}_2(\rho_{XE} \| \id_X \otimes \sigma_E)}}\\
 &\leq\eps^2+\frac12\sqrt{
 2\bigl[|\mathcal Z|-1+\mu(|\mathcal X|-1)\bigr]
 \kappa(\eps)\,2^{-\hmin^\eps(X|E)_\rho}}.
 \label{eq:ordinary}
\end{aligned}\end{equation}
Here $\kappa(\eps) \leq \frac43$ is again the function as in~\eqref{eq:kappa_appendix}.
\end{proposition}

\begin{proof}
Recall from the argument in Section~\ref{subsec:theorem1} that
\begin{equation}\begin{aligned}
 \Delta(\rho,\mathcal H) &\leq \Tr(\rho-R)+ \frac12 \sqrt{ |\mathcal Z| \,\sum_{x,x'}G_{xx'}\langle R_x,R_{x'}\rangle_{\sigma_E}}
\end{aligned}\end{equation}
with the collision matrix $G$ defined by $ G_{xx'} = \Pr_S[h_S(x)=h_S(x')]-\frac1{|\mathcal Z|}$. Notice that, for any vector $|c\rangle = (c_x)_x$, we can expand
\begin{equation}\begin{aligned}
    \langle c|G|c\rangle =  \sum_{x,x'} c_x^* c_{x'} \mathbb{E}_S \left(\delta_{h_S(x),h_S(x')} - \frac{1}{|\mathcal Z|}\right) = \mathbb{E}_S \sum_z \left|\sum_x \left(\delta_{f(x),z}-\frac1{|\mathcal Z|}\right) c_x \right|^2 \geq 0,
\end{aligned}\end{equation}
which shows that $G \geq 0$. This then gives the bound
\begin{equation}\begin{aligned}
\sum_{x,x'}G_{xx'}\langle R_x,R_{x'}\rangle_{\sigma_E} \leq \norm G_\infty \sum_{x}\langle R_x,R_{x}\rangle_{\sigma_E} = \norm G_\infty \norm{R_{XE}}_{I_X \otimes \sigma_E}^2
 \label{eq:generalhash_repeat}
\end{aligned}\end{equation}
that applies to any family of hash functions. 
Crucially, in the case at hand, the assumption of $\mu$-almost 2-universality in equation~\eqref{eq:almostcollision} bounds the off-diagonal entries of $G$ by $\mu/|\mathcal Z|$, which will allow us to estimate its operator norm. 
To see this, put
\begin{equation}
 B=\left(1-\frac{1+\mu}{|\mathcal Z|}\right)\id
       +\frac\mu{|\mathcal Z|}\allones-G,
\end{equation}
where $\allones$ stands for the all-ones matrix. 
The matrix $B$ is symmetric, has zero diagonal, and is entrywise nonnegative. Moreover, $G\geq0$ and $\allones\leq|\mathcal X|\id$ imply
\begin{equation}
 B\leq\left(1-\frac{1+\mu}{|\mathcal Z|}\right)\id
       +\frac\mu{|\mathcal Z|}\allones
 \leq\frac{|\mathcal Z|-1+\mu(|\mathcal X|-1)}{|\mathcal Z|}\id.
\end{equation}
For a symmetric entrywise nonnegative matrix, the Perron--Frobenius theorem states that its spectral radius is its largest eigenvalue. Thus the same coefficient bounds the absolute value of every eigenvalue of $B$. Substituting the resulting lower bound on $B$ into the expression for $G$ gives
\begin{equation}
 0\leq G\leq
 \frac{2\bigl[|\mathcal Z|-1+\mu(|\mathcal X|-1)\bigr]}{|\mathcal Z|}\id.
 \label{eq:almostmatrix}
\end{equation}

Putting everything together, insert equation~\eqref{eq:almostmatrix} into the general hashing bound in~\eqref{eq:generalhash_repeat}, then invoke Proposition~\ref{prop:entropycomparison} to obtain the bound in terms of $2^{-\hmin^\eps(X|E)_\rho}$.
\end{proof}

\begin{remark}[Ordinary 2-universality]
For $\mu=0$, Proposition~\ref{prop:ordinary} gives
\begin{equation}
 \Delta(\rho,\mathcal H)
 \leq\eps^2+\frac12\sqrt{2\,(|\mathcal Z|-1)\,\kappa(\eps)\,
                         2^{-\hmin^\eps(X|E)_\rho}}.
\end{equation}
Thus Theorem~\ref{thm:qkd} remains valid with the privacy amplification term under the square root multiplied by two. At fixed entropy, smoothing parameter, and secrecy target, the resulting certified key length is at most one bit shorter than that certified by Theorem~\ref{thm:pa}.
\end{remark}

\begin{remark}[Constant collision probability]
If $\Pr_S[h_S(x)=h_S(x')]=c\leq1/|\mathcal Z|$ for every $x\ne x'$, then the sharper bound
\begin{equation}
 \Delta(\rho,\mathcal H)
 \leq\eps^2+\frac12\sqrt{|\mathcal Z|\,(1-c)\,\kappa(\eps)\,
                         2^{-\hmin^\eps(X|E)_\rho}}
\end{equation}
holds, since $G=(1-c)\id+(c-1/|\mathcal Z|)\allones\leq(1-c)\id$. In particular, $c=1/|\mathcal Z|$ recovers Theorem~\ref{thm:pa}. In the QKD bound, the coefficient $2^\ell-1$ is replaced by $2^\ell(1-c)\leq2^\ell$.
\end{remark}
\section{Proof of Theorem~\ref{thm:qkd}}\label{app:qkdproof}

In this appendix we prove Theorem~\ref{thm:qkd}. We follow the security proof of~\cite{tl2017}, to which we refer for the formal description of the protocol, and only reproduce the steps where our argument departs from it. The notation used in this appendix is summarised in Table~\ref{tab:notation}.

{
\renewcommand{\thetheorem}{\ref{thm:qkd}}
\begin{theorem}[Finite-key security]
Fix the round counts $m,k,n$ with $m=k+n$ and the test threshold $\Qtol$. Under the protocol assumptions of~\cite{tl2017}, specialised to ideal BB84 measurements and with a $2^*$-universal family for privacy amplification, an $\eps_{\mathrm{tot}}$-secure final key of length $\ell$ can be extracted whenever
\begin{equation}
 \inf\left\{\eps_{\mathrm{cor}}(t)+\eps_\pe(\nu)+
 \sqrt{\frac{2^{\ell-H_*(\nu,t)}}{3}}\;:\;0<\nu<\frac12-\Qtol,\ t\in\mathbb N\right\}
 \leq\eps_{\mathrm{tot}},
\end{equation}
with $H_*(\nu,t)=n[1-h(\Qtol+\nu)]-\leakEC-t$.
\end{theorem}
}

\begin{table}[!htbp]
 \centering
 \small
 \begin{tabular}{ll}
  \toprule
  Symbol & Meaning \\
  \midrule
  $\Pi$ & Uniformly random $k$-subset of the $m$ rounds used for testing \\
  $Z\in\{0,1\}^m$ & Disagreements between Alice's and Bob's outcomes over all $m$ rounds \\
  $\pe$ & Event that the parameter estimation test passes \\
  $\pass$ & Event that parameter estimation and verification both pass \\
  $\Omega$ & Event that the key rounds contain at least $n(\Qtol+\nu)$ errors \\
  $\rho_{\wedge A}$ & Restriction of $\rho$ to the event $A$, with trace $\Pr[A]_\rho$ \\
  $\hmin(X|E\wedge A)_\rho$ & Min-entropy evaluated on $\rho_{\wedge A}$ \\
  $\sigma$ & State after measurement and parameter estimation \\
  $\sigma'$ & State after verification, before privacy amplification \\
  $E$ & Adversary's system with all public seeds and transcripts except $S$ \\
  $S$ & Privacy amplification seed \\
  $K$ & Alice's final key $h_S(X)$ \\
  $\eps_\pe(\nu)$ & Worst-case sampling failure probability, equation~\eqref{eq:peprob} \\
  $H_*(\nu,t)$ & $n[1-h(\Qtol+\nu)]-\leakEC-t$, equation~\eqref{eq:hstar} \\
  \bottomrule
 \end{tabular}
 \caption{\textbf{Notation used in Appendix~\ref{app:qkdproof}.}}
 \label{tab:notation}
\end{table}

\begin{proposition}\label{prop:sampling}
Let $Z=(Z_1,\ldots,Z_m)$ be binary random variables and let $\Pi$ be a uniformly random $k$-subset of $[m]$, independent of $Z$. Then
\begin{equation}
 \Pr\left[\frac1k\sum_{i\in\Pi}Z_i<\Qtol\ \text{and}\ \frac1n\sum_{i\notin\Pi}Z_i\geq\Qtol+\nu\right]\leq\eps_\pe(\nu).
\end{equation}
\end{proposition}
\begin{proof}
By the law of total probability and the independence of $\Pi$ and $Z$, we have
\begin{equation}\begin{aligned}
 &\Pr\left[\frac1k\sum_{i\in\Pi}Z_i<\Qtol\ \text{and}\ \frac1n\sum_{i\notin\Pi}Z_i\geq\Qtol+\nu\right]\\
 &\qquad=\sum_{z\in\{0,1\}^m}\Pr[Z=z]\,\Pr_\Pi\left[\frac1k\sum_{i\in\Pi}z_i<\Qtol\ \text{and}\ \frac1n\sum_{i\notin\Pi}z_i\geq\Qtol+\nu\right]
 \leq\eps_\pe(\nu),
\end{aligned}\end{equation}
where the inequality follows from the definition~\eqref{eq:peprob}.
\end{proof}
This replaces the Serfling bound of~\cite[Lemma~6]{tl2017}; as there, no assumption is made on the distribution of $Z$.

\begin{proposition}\label{prop:hmax}
Let $\sigma$ be the state after measurement and parameter estimation. For any $\nu\in(0,\frac12-\Qtol)$ such that $\eps_\pe(\nu)<\Pr[\pe]_\sigma$, we have
\begin{equation}
 \hmax^{\sqrt{\eps_\pe(\nu)}}(X|Y\wedge\pe)_\sigma\leq n\,h(\Qtol+\nu).
\end{equation}
\end{proposition}
\begin{proof}
Let $\Omega$ be the event that $X_i\neq Y_i$ for at least $n(\Qtol+\nu)$ indices $i\in[n]$, and let $Z\in\{0,1\}^m$ record the rounds in which Alice's and Bob's outcomes disagree. Every round is measured in a basis chosen independently of $\Pi$~\cite[Sec.~3.3]{tl2017}, so $Z$ is independent of $\Pi$, and Proposition~\ref{prop:sampling} gives
\begin{equation}
 \Pr[\pe\wedge\Omega]_\sigma\leq\eps_\pe(\nu).
\end{equation}
Since $\eps_\pe(\nu)<\Pr[\pe]_\sigma$, we can apply~\cite[Lemma~7]{tl2017} to the subnormalised branch $\sigma_{XY\wedge\pe}$ where Alice and Bob pass parameter estimation, to show the existence of a state $\tilde\sigma_{XY}$ with $\Pr[\Omega]_{\tilde\sigma}=0$ and $P(\sigma_{XY\wedge\pe},\tilde\sigma_{XY})\leq\sqrt{\eps_\pe(\nu)}$. Hence
\begin{equation}
 \hmax^{\sqrt{\eps_\pe(\nu)}}(X|Y\wedge\pe)_\sigma\leq\hmax(X|Y)_{\tilde\sigma}\leq n\,h(\Qtol+\nu),
\end{equation}
where the first inequality follows since the smoothing takes an infimum, and the second from the argument in~\cite[Proposition~8, Eqs.~(83)--(89)]{tl2017}, which uses only $\Pr[\Omega]_{\tilde\sigma}=0$ and $\Qtol+\nu\leq\frac12$.
\end{proof}

\begin{proposition}\label{prop:hmin}
Let $\sigma'$ be the state after error correction and verification, and let $\nu\in(0,\frac12-\Qtol)$ satisfy $\eps_\pe(\nu)<\Pr[\pass]$. Then
\begin{equation}
 \hmin^{\sqrt{\eps_\pe(\nu)}}(X|E\wedge\pass)_{\sigma'}\geq H_*(\nu,t),
\end{equation}
where $H_*(\nu,t)=n[1-h(\Qtol+\nu)]-\leakEC-t$.
\end{proposition}
\begin{proof}
Since $\Pr[\pass]\leq\Pr[\pe]_\sigma$, Proposition~\ref{prop:hmax} applies. Combining it with the entropic uncertainty relation~\cite[Corollary~5]{tl2017} for ideal BB84 measurements, $\overline{c}=\frac12$, and discarding Bob's test string, we have
\begin{equation}
 \hmin^{\sqrt{\eps_\pe(\nu)}}(X|E\wedge\pe)_\sigma\geq n[1-h(\Qtol+\nu)].
\end{equation}
During error correction Alice reveals a syndrome of at most $\leakEC$ bits and a tag of $t$ bits. By the chain rule~\cite[Eq.~(19)]{tl2017} each lowers the smooth min-entropy by at most its length, and restricting further to acceptance of verification cannot lower it~\cite[Lemma~10]{tl2017}; the smoothing radius remains admissible since $\eps_\pe(\nu)<\Pr[\pass]$. This is the chain of inequalities in~\cite[Eqs.~(97)--(102)]{tl2017}, and gives the claim.
\end{proof}

\begin{proof}[Proof of Theorem~\ref{thm:qkd}]
By~\cite[Lemma~1]{tl2017}, it suffices to bound the correctness and secrecy errors separately. By~\cite[Theorem~2]{tl2017}, the verification hash of length $t$ gives correctness error $\eps_{\mathrm{cor}}(t)=2^{-t}$. It therefore remains to show that for every $\nu\in(0,\frac12-\Qtol)$,
\begin{equation}
 \frac12\norm{\omega_{KSE\wedge\pass}-\pi_K\otimes\omega_{SE\wedge\pass}}_1\leq\eps_\pe(\nu)+\sqrt{\frac{2^{\ell-H_*(\nu,t)}}{3}}.
 \label{eq:secrecy-appendix}
\end{equation}

First consider the case $\eps_\pe(\nu)<\Pr[\pass]$. Then $\sqrt{\eps_\pe(\nu)}<\sqrt{\Tr\sigma'_{\wedge\pass}}$, so we can apply Theorem~\ref{thm:pa} to $\sigma'_{XE\wedge\pass}$ with smoothing parameter $\sqrt{\eps_\pe(\nu)}$ and $|\mathcal Z|=2^\ell$. The left-hand side of~\eqref{eq:secrecy-appendix} is exactly $\Delta(\sigma'_{\wedge\pass},\mathcal H)$, since the secrecy quantity includes the seed (Appendix~\ref{app:conventions}). Using Proposition~\ref{prop:hmin} and for simplicity bounding $\kappa(\sqrt{\eps_\pe(\nu)})\leq\frac43$ and $2^\ell-1\leq2^\ell$, we obtain
\begin{equation}
 \Delta(\sigma'_{\wedge\pass},\mathcal H)\leq\eps_\pe(\nu)+\frac12\sqrt{\frac43\,2^{\ell-H_*(\nu,t)}}=\eps_\pe(\nu)+\sqrt{\frac{2^{\ell-H_*(\nu,t)}}{3}}.
\end{equation}

Now consider the case $\Pr[\pass]\leq\eps_\pe(\nu)$. Both operators in~\eqref{eq:secrecy-appendix} have trace $\Pr[\pass]$, so
\begin{equation}
 \frac12\norm{\omega_{KSE\wedge\pass}-\pi_K\otimes\omega_{SE\wedge\pass}}_1
 \leq\frac12\left(\Tr\omega_{KSE\wedge\pass}+\Tr[\pi_K\otimes\omega_{SE\wedge\pass}]\right)
 =\Pr[\pass]\leq\eps_\pe(\nu).
\end{equation}
In neither case is the secrecy bound divided by $\Pr[\pass]$.

Adding $\eps_{\mathrm{cor}}(t)$ and taking the infimum over $\nu$ bounds the total security error for fixed $t$. Since the square-root term diverges as $t\to\infty$, the infimum over $t\in\mathbb N$ in~\eqref{eq:qkd} is attained, and choosing this $t$ as the verification length proves the theorem.
\end{proof}
\section{Exact sampling tails}\label{app:sampling}
The following lemma justifies why evaluating $\eps_{\mathrm{pe}}(\nu)$ numerically can be done in constant queries, instead of searching over all binary error patterns $z \in \{0,1\}^m$.
\begin{lemma}[Exact test and sampling tails]\label{lem:exact-tails}
Let $m=k+n$ with $k,n\geq1$, let $0<\Qtol<\Qtol+\nu<1/2$, and put
\begin{equation}
 a=\lceil k\Qtol\rceil-1,\qquad b=\lceil n(\Qtol+\nu)\rceil.
 \label{eq:appendix-test-count}
\end{equation}

For a uniform $k$-subset $\Pi$, independent of a binary error pattern $z$, the worst-case joint probability in equation~\eqref{eq:peprob} is
\begin{equation}
 \begin{aligned}
 \eps_\pe(\nu)
 &=\Pr[\operatorname{Hyper}(m,a+b,k)\leq a]\\
 \end{aligned}
 \label{eq:hypergeom}
\end{equation}
Here $\operatorname{Hyper}(m,s,k)$ counts sampled errors when $k$ positions are drawn without replacement from $m$ positions containing $s$ errors. The maximum is attained by any pattern with $a+b$ errors.
\end{lemma}
\begin{proof}
Fix a pattern of weight $s$. Its sampled error count is $T_s\sim\operatorname{Hyper}(m,s,k)$, and the unsampled count is $s-T_s$. Thus the joint probability is
\begin{equation}
 p(s)=\Pr[T_s\leq\min\{a,s-b\}].
\end{equation}
Couple the variables using one uniform $k$-subset $\Pi$ and setting $T_s=|\Pi\cap\{1,\ldots,s\}|$. Then $T_s\leq T_{s+1}\leq T_s+1$ pointwise. For $s<a+b$, the event $T_s\leq s-b$ implies $T_{s+1}\leq s+1-b$, so $p(s)\leq p(s+1)$. For $s\geq a+b$, the threshold is $a$ and the event $T_{s+1}\leq a$ implies $T_s\leq a$, so $p(s+1)\leq p(s)$. The maximum is therefore attained at $s=a+b$, where the hypergeometric tail gives equation~\eqref{eq:hypergeom}. Averaging over patterns independent of $\Pi$ preserves the bound.
\end{proof}

\section{Numerical optimisation}\label{app:numerics}

This appendix specifies the implementation of Section~\ref{sec:methods-numerics}, the code for which is available at~\cite{leftoverscode2026}.

For fixed protocol inputs (the total number of signals $m$, the expected QBER $Q$,
the security and completeness parameters $\eps_{\mathrm{tot}}$ and
$\eps_{\mathrm{com}}$, and the error correction leakage parameters), the certified
key length is
\begin{equation}
 \ell_{\mathrm{new}}=
 \max\left\{0,\left\lfloor H_*(\nu,t)+\log_2 3+
 2\log_2\!\left[\eps_{\mathrm{tot}}-2^{-t}-\eps_\pe(\nu)\right]\right\rfloor\right\},
 \label{eq:numerical-new-length}
\end{equation}
which we maximise over the remaining free parameters $t$, $k$ and $\nu$. Here
$H_*$ and $\eps_\pe$ also depend on $k$, through $n = m-k$ and $Q_{\mathrm{tol}}$;
we suppress this as in Theorem \ref{thm:qkd}. The optimal $t$ has a closed
form, and we optimise over $k$ and $\nu$ by a grid search, as described below.

\subsection{Integer parameters and leakage}\label{sec:reconciliation-model}

For exact tails, at each candidate $k$ the code inverts the binomial survival function at the assigned test-abort probability $\eps_{\mathrm{com}}/2$ to obtain $a$. It checks that the tail at $a$ is at most $\eps_{\mathrm{com}}/2$ and that the tail at $a-1$ exceeds $\eps_{\mathrm{com}}/2$, adjusting the integer quantile when necessary. For analytic tails, an integer bisection instead finds the smallest $a+1$ satisfying equation~\eqref{eq:chernoff-completeness} with budget $\eps_{\mathrm{com}}/2$; the comparison is evaluated in the logarithmic domain. With $n=m-k$, the search uses $\nu=j/n$ for positive integers $j$ and computes
\begin{equation}
 \Qtol=\frac{a+1}{k},\qquad
 b=\left\lceil\frac{n(a+1)}{k}\right\rceil+j.
\end{equation}
Thresholds in $(a/k,(a+1)/k]$ give the same accepted test counts; we fix the upper endpoint rather than optimise within this interval. At $Q=0\%$, the choice is $a=0$ and $\Qtol=1/k$. The test rejects equality at $k\Qtol$, so the largest accepted count is $a=\lceil k\Qtol\rceil-1$, rather than $\lfloor k\Qtol\rfloor$. The strict constraint $\Qtol+\nu<1/2$ is enforced as $2n(a+1)+2kj<kn$. These integer counts are retained for the tail evaluations instead of being reconstructed from floating-point rates.

The integer leakage used in the calculation is
\begin{equation}
 \leakEC=\left\lceil\xi_1nh(Q)
 -\xi_2\sqrt{nV(Q)}\,\Phi^{-1}(\eps_{\mathrm{com}}/2)\right\rceil.
 \label{eq:integer-leakage}
\end{equation}
The coefficients are chosen as specified in Section~\ref{sec:methods-numerics}. At $Q=0\%$, we assume zero leakage. The leakage depends on $n$ and the physical $Q$, not on $\Qtol$.

\subsection{Objectives and verification length}

For the key rate and key length curves, the new LHL calculation uses the security test
\begin{equation}
 2^{-t}+\eps_\pe(\nu)+
 \sqrt{\frac{2^{\ell-H_*(\nu,t)}}{3}}
 \leq\eps_{\mathrm{tot}}.
 \label{eq:numerical-new-security}
\end{equation}
For positive remaining budget $\eps_{\mathrm{tot}}-2^{-t}-\eps_\pe(\nu)$, the objective is
\begin{equation}
 \ell_{\mathrm{new}}=
 \max\left\{0,\left\lfloor H_*(\nu,t)+\log_2 3+
 2\log_2\!\left[\eps_{\mathrm{tot}}-2^{-t}-\eps_\pe(\nu)\right]\right\rfloor\right\}.
 \label{eq:numerical-new-length2}
\end{equation}
The expression is evaluated directly without exponentiating the entropy. For the old LHL, the security test is
\begin{equation}
 2^{-t}+2\sqrt{\eps_\pe(\nu)}+
 \frac12\sqrt{2^{\ell-H_*(\nu,t)}}\leq\eps_{\mathrm{tot}},
\end{equation}
giving
\begin{equation}
 \ell_{\mathrm{old}}=\max\left\{0,
 \left\lfloor H_*(\nu,t)+
 2\log_2\!\left[\eps_{\mathrm{tot}}-2^{-t}-2\sqrt{\eps_\pe(\nu)}\right]+2\right\rfloor\right\}.
\end{equation}
Candidates with nonpositive remaining budget are discarded; a search with no positive-key candidate reports zero. The analytic-tail curves use the same objectives, with equation~\eqref{eq:serfling} for parameter estimation and equation~\eqref{eq:chernoff-completeness} to set $\Qtol$.

For fixed $k,\Qtol,\nu,\leakEC$, let $A=\eps_{\mathrm{tot}}-\eps_\pe(\nu)$ for the new LHL and $A=\eps_{\mathrm{tot}}-2\sqrt{\eps_\pe(\nu)}$ for the old LHL. Since $H_*(\nu,t)=H_*(\nu,0)-t$, both objectives before rounding are monotone in $x(A-x)^2$, where $x=2^{-t}$. For $A>0$, this expression is maximised at $x=A/3$, giving, for the new LHL,
\begin{equation}
 t_{\mathrm{cont}}=\log_2\frac{3}{\eps_{\mathrm{tot}}-\eps_\pe(\nu)}.
 \label{eq:tag-optimum}
\end{equation}
For the old LHL, replace $\eps_\pe(\nu)$ in equation~\eqref{eq:tag-optimum} by $2\sqrt{\eps_\pe(\nu)}$. The code compares the positive integer neighbours of $t_{\mathrm{cont}}$ and keeps the one giving the larger objective before key-length rounding; ties select the shorter tag. This optimises $t$ without adding a search dimension and is also used for one-bit feasibility.

\subsection{Grid refinement for key lengths}

For each block size $m$, we maximise the key length numerically over the
test-sample size $k$ and the parameter $\nu$. We start from a coarse grid over
the admissible range, with $\nu$ restricted to multiples of $1/n$, where $n$ is
the number of key bits remaining after sampling. At each grid point we set $t$
to its analytic optimum and compute the resulting key length. We then refine
the grid around the best candidate and repeat until the grid spacing reaches
integer resolution or a fixed number of passes is exhausted, never discarding a
better candidate found earlier. Since every admissible choice of parameters
yields a valid key length, any suboptimality of the search only makes the
reported rates conservative. The REAT curves are computed with the code
of~\cite{gohtan2026}.

Figures~\ref{fig:rates} and~\ref{fig:sampling-comparison} use the same 147 integer block sizes: the union of 81 logarithmically spaced values from $10^3$ to $2\times10^4$, 41 from $5000$ to $10^5$, 25 from $10^5$ to $10^7$, and $10^4$, after rounding and deduplication. Figure~\ref{fig:satellite}(a) uses $m\in\{1000,2000,3000,4000,5000\}$ and $Q=0\%,0.025\%,\ldots,2\%$, with both grid resolutions at each point. Table~\ref{tab:rates} reads selected rates from the same datasets as Fig.~\ref{fig:rates}. The REAT data are pre-computed using the code of~\cite{gohtan2026} and not part of this optimisation.

\subsection{Minimum block size for one-bit extraction}

For Fig.~\ref{fig:satellite}(b), we apply Theorem~\ref{thm:pa} directly at $\ell=1$, retaining $2^\ell-1=1$ and the full coefficient~\eqref{eq:kappa}:
\begin{equation}
 2^{-t}+\eps_\pe(\nu)+\frac12\sqrt{\kappa\!\left(\sqrt{\eps_\pe(\nu)}\right)}\,2^{-H_*(\nu,t)/2}
 \leq\eps_{\mathrm{tot}}.
 \label{eq:one-bit}
\end{equation}
Here $\kappa(\sqrt{\eps_\pe})=(1-\eps_\pe)(1+3\eps_\pe)$ throughout the feasible range. The search covers all integer $2\leq k<m$, positive integer $j$ and corresponding $\nu=j/(m-k)$ satisfying the strict half-error constraint. It optimises $t$ as above. Starting at $m=256$, doubling finds a feasible upper block size. A recursive search then examines the entire integer interval from $3$ to that upper size, always visiting the lower half first. There are no admissible test sizes for $m<3$. The algorithm excludes intervals using the following necessary conditions, rather than assuming that integer-rounded feasibility is monotone in $m$.

Write $H_0=H_*+t$ for the entropy before verification. If $t_0$ is the optimal tag at zero sampling error, any feasible candidate must satisfy
\begin{equation}
 H_0\geq H_{0,\min}\coloneqq t_0-\log_2\!\left[4(\eps_{\mathrm{tot}}-2^{-t_0})^2\right],
 \qquad \eps_\pe<\eps_{\mathrm{tot}}.
\end{equation}
For a block-size interval $[L,U]$ and a fixed $k<U$, put $n_-=\max\{1,L-k\}$ and $n_+=U-k$. The count $a$ and $\Qtol$ depend only on $k$ and $Q$. Let $r(n)$ denote the leakage~\eqref{eq:integer-leakage}. Since $r(n)$ and $n[1-h(\Qtol+j/n)]$ are nondecreasing in $n$ in the admissible range,
\begin{equation}
 H_0\leq n_+[1-h(\Qtol+j/n_+)]-r(n_-).
 \label{eq:threshold-entropy-upper}
\end{equation}
The right-hand side decreases with $j$. Integer bisection finds the largest $j_{\max}$ for which it can reach $H_{0,\min}$, subject to the half-error constraint at $n_+$. If no positive $j$ survives, this $k$ is excluded.

For a surviving $k$, every candidate in the interval has sampling probability at least
\begin{equation}
 \Pr\!\left[\operatorname{Hyper}\!\left(
 k+n_-,\min\{k+n_-,\,a+\lceil n_+\Qtol\rceil+j_{\max}\},k\right)\leq a\right].
 \label{eq:threshold-tail-lower}
\end{equation}
At fixed total errors this cumulative probability increases with population size, and at fixed population size it decreases with total errors. Taking these two optimistic endpoints is therefore valid even when they cannot be attained together. The code excludes $k$ when the lower bound exceeds $\eps_{\mathrm{tot}}$. Comparisons allow an absolute entropy margin of $10^{-10}$ and a relative probability margin of $10^{-10}$ to avoid pruning close candidates through rounding.

An interval is excluded if no $k$ remains; otherwise it is bisected until a single block size is reached. At that size the code evaluates every surviving $k$ and every $1\leq j\leq j_{\max}$ with the full test~\eqref{eq:one-bit}. The first feasible size is $M_0$ on the stated lattice, and the exclusion intervals account for every smaller size. This procedure is run at the 81 QBERs above for each security target in Fig.~\ref{fig:satellite}(b), giving 243 witnesses. 

\end{document}